\documentclass[11pt,a4paper]{article}

\usepackage[margin=2.5cm]{geometry}

\usepackage{amsmath}
\usepackage{amssymb}
\usepackage{amsthm}
\usepackage{mathtools}
\usepackage{graphicx}
\usepackage{enumitem}

\usepackage[hidelinks]{hyperref}

\usepackage{microtype}

\usepackage{tikz}

\usepackage{algorithmicx}
\usepackage{algorithm}
\usepackage[noend]{algpseudocode}

\usetikzlibrary{snakes}

\theoremstyle{plain}
\newtheorem{theorem}{Theorem}

\newtheorem{lemma}[theorem]{Lemma}
\newtheorem{corollary}[theorem]{Corollary}

\theoremstyle{definition}
\newtheorem{definition}[theorem]{Definition}

\theoremstyle{remark}

\newtheorem{fact}[theorem]{Fact}
\newtheorem{remark}[theorem]{Remark}
\newtheorem*{remark*}{Remark}

\newcommand{\OPT}{\mathsf{OPT}}
\newcommand{\NN}{{\mathbb N}}
\DeclareMathOperator{\cost}{cost}

\newcommand{\pare}[1]{\left( #1 \right)} 

\newcommand{\abs}[1]{\left| #1 \right|} 
\newcommand{\floor}[1]{\left\lfloor #1 \right\rfloor} 

\newcommand{\setbuild}[2]{\left\{ #1 \middle| #2 \right\}} 

\title{
From Trees to Binary Trees: Approximation Algorithms and a PTAS
}

\author{\\Fabian Weisbrodt, Florian Werner\\[5mm]
Technical University of Munich, Munich, Germany\\[5mm]
fabian.weisbrodt@tum.de, florian.werner@cit.tum.de
}

\date{}
\begin{document}

\maketitle

\begin{abstract}
\noindent Kellerhals, Krebs, Nichterlein, and Schmid
introduced the optimization problem \textsc{Tree in Binary Tree} (TBT),
where, given a tree $T$, the task is to find a binary tree $S^*$ with $V(T)=V(S^*)$ that minimizes $\sum_{uv\in E(T)} \operatorname{dist}_{S^*}(u,v)$. They also introduced the variant \textsc{Tree in Binary Tree With Steiner Nodes} (TBT-SN), where the condition $V(T)=V(S^*)$ is relaxed to $V(T)\subseteq V(S^*)$. 
We give a simple linear-time $\frac{8}{3}$-approximation algorithm for TBT and a polynomial-time approximation scheme (PTAS) for TBT-SN. 
Our main building idea in both algorithms is to iteratively consider subtrees of the tree $T$ which contain all vertices below some vertex $v$ and construct a solution where $v$ has degree 1.
For the PTAS, we cut $T$ into smaller trees, solve these instances, modify the obtained solutions, and combine them to achieve a solution for the original instance.\\

\end{abstract}


\thispagestyle{empty} \newpage{}

\setcounter{page}{1}

\setlength\parindent{0pt}

\section{Introduction}
Many algorithmic problems on graphs are significantly easier to solve when we restrict the input to trees. This simple observation motivates approximating a general graph metric, i.e. the distance function of a graph, by a tree metric. 
One highlight is Bartal's \cite{Bartal} idea to 
consider the expected distortion when approximating a metric by a tree metric.
Fakcharoenphol, Rao, and Talwar~\cite{FAKCHAROENPHOL2004485} 
showed that a metric can be approximated by tree metrics with expected distortion $(O(\log n))$. 
A natural next step is to restrict the properties of the trees behind the tree metric, especially introducing bounds on the degree. We focus on this step from general trees to trees of bounded degree, more specifically, binary trees.
Recently, Kellerhals, Krebs, Nichterlein, and Schmid~\cite{paper} introduced the optimization problem \textsc{Tree in Binary Tree} (TBT) in the context of network embedding problems, on which we focus in this paper.
$$
\begin{aligned}
&\textsc{Tree in Binary Tree (TBT)}\\[0.5em]
&\textbf{Input: } \text{A tree } T\\
&\textbf{Task: }
\text{Find a binary tree } S^*
\text{ with } V(T) = V(S^*) \text{ that minimizes} \sum_{uv \in E(T)} \operatorname{dist}_{S^*}(u,v).
\end{aligned}
$$
Kellerhals, Krebs, Nichterlein, and Schmid also introduced the variant
\textsc{Tree in Binary Tree with Steiner Nodes (TBT-SN)}, where the condition $V(T) = V(S^*)$ is replaced with  $V(T) \subseteq V(S^*)$. Intuitively, TBT-SN also allows Steiner nodes.
They
gave a linear-time 4-approximation algorithm for TBT, which first calculates a 3-approximate solution for TBT-SN and then eliminates the Steiner nodes.\\
In \cite{paper}, \textsc{Tree in Binary Tree} is motivated as an embedding problem. 
We highlight two related problems. In the first, trees are embedded into binary trees under a different objective function; in the second, trees are embedded into paths, but the objective function is the same as in TBT. Specifically: \\
First, in \cite{monien1985complexity}, Monien showed that, given an input tree $T$ of height 3 (diameter 6), finding a binary tree $S$ with $V(T)\subseteq V(S)$ such that $\max_{uv\in E(T)} \operatorname{dist}_S(u,v)$ is minimal is NP-hard. In embedding problems, this objective function is called \emph{bandwidth}. The computational complexity of TBT and TBT-SN is unknown; however, it seems highly unlikely that replacing the objective function bandwidth with the objective function in TBT makes the problem easier.\\
Second, given a tree $T$ as input, the task is to find a path $P$ with $V(T)\subseteq V(P)$ such that $\sum_{uv\in E(T)} \operatorname{dist}_P(u,v)$ is minimal. In \cite{GoldbergKlipker1976},
Goldberg and Klipker showed that an optimal solution can be found in polynomial time. The running time was later improved by Shiloach in \cite{Shiloach} and by Chung in \cite{CHUNG198443}.
In contrast to trees, given a bipartite graph $G$, computing a path $P$ with $V(P)\subseteq V(G)$ minimizing  $\sum_{uv\in E(G)} \operatorname{dist}_P(u,v)$ is NP-complete
\cite{diaz2002survey}.\\
TBT is a specialization of the well-studied \textsc{network design problem} \cite{mejri2023survey}. Given a weighted graph, the classical network design problem asks for a connected subgraph which optimizes a cost function related to communication cost under budget constraints, which can, for example, limit the number of edges or the total edge weight. A common objective function is the sum of all pairwise distances in the connected subgraph.
Relaxing the condition that a feasible solution must be a subgraph and further allowing a larger vertex set
leads to problems involving Steiner nodes, including the classical Steiner tree problem, which is NP-complete \cite{dreyfus, karp}. 
Liers, Martin, and Pape~\cite{LIERS201685} showed that computing a minimal binary Steiner tree is NP-complete.
Note that in TBT, the optimal binary tree is of course not a subgraph of the input tree $T$ if $T$ has a vertex of degree at least 4. However, one can consider a weighted complete graph $K$ with the same vertex set as $T$, where the weight on each edge $xy$ is the distance between the vertices $x$ and $y$ in $T$. 
A classical variant of the \textsc{network design problem} is the \textsc{Minimum Routing Cost Tree problem} (MRCT). Given a graph, the task is to find a subtree that minimizes the sum of all pairwise distances. Up to a scaling factor, the problem is also known as the minimum average distance (MAD) tree.
MRCT was introduced by Hu in \cite{hu} under the name \textsc{Optimum Distance Spanning Tree} problem. According to Hu, the problem was proposed to him by Francesco Maffioli. Hu also introduced a more general version for weighted graphs.
Johnson, Lenstra, and Rinnooy Kan~\cite{repec:ags:eureia:272157} showed that MRCT is NP-hard.
The authors in \cite{wu} gave a polynomial-time approximation scheme (PTAS), where the
input graph is allowed to have nonnegative weights.
Dahlhaus, Dankelmann, and Ravi~\cite{DAHLHAUS2004255} gave a linear-time algorithm on interval graphs.\\

The paper is structured as follows: Section 2 states our main results; Section 3 introduces the necessary notation. Sections 4 and 5 describe our algorithms for TBT and TBT-SN, respectively, with the corresponding analysis in Sections 6 and 7.

\section{Main Results}

Our results are the following two theorems.

\begin{theorem}
There exists a linear-time algorithm that, given a tree $T$ as an instance of TBT, determines a binary tree $S$ with $V(T) = V(S)$ such that 
$$ \sum_{uv \in E(T)} \operatorname{dist}_{S}(u,v) \le \frac{8}{3} \sum_{uv \in E(T)} \operatorname{dist}_{S^*}(u,v),$$
where $S^*$ is an optimal solution.
\end{theorem}

\begin{theorem}\label{thmptas}
For every $\epsilon > 0$, there exists some
$h \in \mathbb{N}$ with the following property:
Given a tree $T$ with $n$ vertices as an instance of TBT-SN, 
one can determine a binary tree $S$ with $V(T) \subseteq V(S)$ such that 
$$ \sum_{uv \in E(T)} \operatorname{dist}_{S}(u,v) \le (1+\epsilon) \sum_{uv \in E(T)} \operatorname{dist}_{S^*}(u,v),$$
where $S^*$ is an optimal solution, in time $O(n^h)$.
\end{theorem}

While the PTAS for TBT-SN is our main result, it is based on the ideas and techniques presented in the $\frac{8}{3}$-approximation algorithm in Section~4. Furthermore, the analysis of the PTAS in Section~7 also uses several lemmas from Section~6. 

\section{Preliminaries} \label{section_notation}

We use standard graph theory terminology.
The \emph{order} $n(G)$ and the \emph{size} $m(G)$ of a graph $G$ are the numbers of its vertices and edges, respectively. 
A \emph{tree} is a connected, acyclic, and undirected graph. A \emph{rooted} tree is a tree together with one of its vertices designated as the root. Edges of a rooted tree are undirected; the root is only used to define child, parent, ancestor and descendant relations.
The \emph{depth} of a vertex in a rooted tree is its distance from the root.

Let $T$ be a tree with root $r$. 
We denote the vertex set and the edge set by $V(T)$ and $E(T)$, respectively. For a vertex $v \in V(T)$ we denote by $N_T(v)$ its neighborhood and by $d_T(v)$ its degree. Further, we denote by $N^+_T(v)$ the set of its children and define $d^+_T(v) \coloneqq \abs{N^+_T(v)}$.  If the underlying tree is $T$, we sometimes omit the subscript~$T$ in the degree and neighborhood.\\
If $v \notin V(T)$, we define $d_T(v)=0$. This convention will be used when considering Steiner vertices.\\

For any set of vertices $A$, let $T[A]$ denote the subgraph of $T$ induced by $A \cap V(T)$. 
For a vertex $v \in V(T)$, let $T_v$ denote the subtree of $T$ induced by $v$ and all descendants of $v$, that is, the subtree of $T$ rooted at $v$.

For two positive integers $a,b$, let $K_{a,b}$ denote the complete bipartite graph whose two partition classes have sizes $a$ and $b$.
In particular, $K_{1,a}$ is a star with $a+1$ vertices.\\

Given a tree $S$ and an instance $T$ of TBT or TBT-SN, we define 
$$\operatorname{cost}_T(S)=\operatorname{cost}(S)=\sum\limits_{uv\in E(T): u,v\in V(S)} \operatorname{dist}_S(u,v).$$
Unless stated otherwise, $\operatorname{cost}(S)$ refers to the instance $T$, and we omit the subscript $T$ for simplicity.\\
Let $\OPT(T)$ denote the optimal value of the instance $T$; the considered optimization problem (TBT or TBT-SN) is clear from context. 
A tree $S$ is called a (feasible) solution for an instance $T$ of TBT (or TBT-SN) if $S$ is a binary tree and $V(T)=V(S)$ (or $V(T)\subseteq V(S)$).
\\
For a positive integer $k$, let $[k]=\{1,\ldots,k\}$ denote the set of positive integers at most $k$.

\section{A $\mathbf{\frac{8}{3}}$-approximation algorithm for general instances of TBT}

Recall the definition of a complete binary tree, which is central for the construction in the algorithm.
\begin{definition} \label{definition_complete_binary_tree}
    We call a binary tree with $n$ vertices complete if it is isomorphic to the tree with vertices $u_1, \ldots, u_n$ with root $u_1$ in which the vertex $u_i$ has the parent $u_{\floor{i/2}}$ for $2 \leq i \leq n$.
\end{definition}
\noindent Our algorithm builds a binary tree $S$ using a bottom-up approach. The general idea is to build a solution $S^v$ of $T_v$ step by step where the root $v$ of $S^v$, which coincides with the root $v$ of $T_v$, has at most one child. The iteration works as follows.\\
Let $v \in V(T)$ be a vertex of the instance $T$ with children $u_1,\ldots, u_d$. Let $S^{u_1},\ldots, S^{u_d}$ be the solutions of the instances $T_{u_1},\ldots, T_{u_d}$ which the algorithm built before. By construction, $u_i$ has at most one child in $S^{u_i}$  which we call $w_i$, if it exists. We build a complete binary tree using the vertices $u_1,\ldots, u_d$ rooted in $u_1$ and connect $u_1$ to $v$. We call the resulting tree $S^v_+$. We add $S^{u_1}-u_1,\ldots, S^{u_d}-u_d$ to $S^v_+$ via edges between the $w_i$'s and leaves of $S^v_+$ such that the tree remains binary and $w_i$ is a descendant of $u_i$. We call this new tree $S^v$, which is a solution of $T_v$. Note that by construction $v$ has at most one child in $S^v$. Algorithm \ref{algorithm_trees_to_binary_trees} below gives the details.

\begin{algorithm} 
    \caption{Tree to Binary Tree
}
    
    \label{algorithm_trees_to_binary_trees}
    \renewcommand{\algorithmicrequire}{\textbf{Input:}}
    \renewcommand{\algorithmicensure}{\textbf{Output:}}
    \begin{algorithmic}[1]
    \Require tree $T$
    \Ensure binary tree $S$ with $V(S) = V(T)$     
        \While{there exists an unprocessed vertex $v \in V(T)$}
            \State \parbox[t][0pt]{435pt}{let $v \in V(T)$ be not yet processed with children $u_1, \ldots, u_{d}$ such that the subtrees\\ $S^{u_1},\ldots, S^{u_d}$ have already been build and $d^+_T(u_1) \leq \ldots \leq d^+_T(u_d)$} \label{algorithm_8over3_choosing_v_and_sorting_children}
            \State let $w_i$ be the child of $u_i$ in $S^{u_i}$, if it exists, for all $i=1,\ldots ,d$
            \State create the subtree $S^v_+$ as follows:
            \State \hspace{\algorithmicindent} delete the edges between $u_i$ and $w_i$ for all $i = 1, \ldots, d$ \label{algorithm_8over3_start_creating_Sv+}
            \State \hspace{\algorithmicindent} let $v$ be the root of $S^v_+$ with a single child $u_1$  \label{algorithm_8over3_connecting_u1_to_v}
            \State \hspace{\algorithmicindent} \parbox[t][0pt]{416pt}{starting from $u_1$ add the vertices $u_2, \ldots, u_d$ to create a complete binary tree, i.e.\\ the children of $u_i$ in $S^v_+$ are $u_{2i}$ and $u_{2i+1}$ if they exist}  \label{algorithm_8over3_end_creating_Sv+}
            \ForAll{$i = d, \ldots, 1$} \label{algorithm_8over3_loop_adding_subtrees}
                \State consider the subtree $(S^v_+)_{u_i}$ of $S^v_+$ which contains $u_i$ and all its descendants in $S^v_+$
                \State add an edge from $w_i$ to some leaf of $(S^v_+)_{u_i}$ which currently has a degree $<3$ \label{algorithm_8over3_connecting_subtrees}
            \EndFor
            \State define $S^v$ as the tree created by connecting the subtrees $S^{u_i}-u_i$ to $S^v_+$
        \EndWhile
    \end{algorithmic}
\end{algorithm}

\begin{remark}
   To improve readability, we omit explicit existence conditions in Algorithm \ref{algorithm_trees_to_binary_trees}. Note that $S^{u_i}$ can consist only of the vertex $u_i$, implying $w_i$ may not exist.
   Operations involving the vertex $w_i$ are only performed if $w_i$ exists. Similarly, $S^{u_i}-u_i$ is not reattached if it is empty.
We use the same convention in the following algorithms.
\end{remark}

\noindent We briefly discuss the correctness and runtime of Algorithm~\ref{algorithm_trees_to_binary_trees}.
\begin{fact} \label{fact_complete_binary_tree}
    Every complete binary tree $T$ with $n$ vertices has at least $\frac{n}{2}$ leaves. Additionally, for any vertex $v \in V(T)$ the subtree $T_v$ is also a complete binary tree.
\end{fact}
\begin{remark}
    Consider line \ref{algorithm_8over3_connecting_subtrees} of Algorithm~\ref{algorithm_trees_to_binary_trees} and fix the iteration in which the subtree $S^{u_i} - u_i$ is added to $S^v_+$ for some $i = 1, \ldots, d$. According to Fact \ref{fact_complete_binary_tree}, there always exists a leaf in $(S^v_+)_{u_i}$ which currently has a degree less than $3$ in $S_v(u_i)$.\\
    For the same reason, we iterate through the leaves in line \ref{algorithm_8over3_loop_adding_subtrees} 
    in reverse order to ensure that we do not accidentally
    ``overbook''
    a leaf. If we iterated through $u_1, \ldots, u_d$ in an arbitrary order, we might assign $u_1$ and $u_2$ both to some leaf $u_i$ with $i \geq 3$ and would not be able to connect $u_i$ with the subtree $S^{u_i}$ as intended in line \ref{algorithm_8over3_connecting_subtrees}.
\end{remark}
\begin{remark}\label{remark_runtime}
    Every vertex $u \in V(T)$ is considered at most once when creating the tree $S^u_+$ and $S^{v}_+$ for its parent $v$ in lines \ref{algorithm_8over3_start_creating_Sv+}-\ref{algorithm_8over3_end_creating_Sv+} and once when connecting $S^{u}$ to $S^v_+$ in line \ref{algorithm_8over3_connecting_subtrees}. Additionally, we only require the children $u_1, \ldots, u_d$ of the currently considered vertex $v$ to be sorted in line \ref{algorithm_8over3_choosing_v_and_sorting_children}. Hence, Algorithm~\ref{algorithm_trees_to_binary_trees} has a runtime in $O(n \log(n))$.
\end{remark}
\begin{remark} \label{remark_linear_runtime}
    We can adapt Algorithm~\ref{algorithm_trees_to_binary_trees} to get a linear runtime. The general idea is that we divide the vertices $u_1, \ldots, u_d$ into two sets: one containing the vertices with more children than average, the other containing those with at most the average number of children. However, we retain our current presentation of Algorithm~\ref{algorithm_trees_to_binary_trees} to better convey the intuition behind it and show that the linear runtime adaptation maintains the same approximation in Appendix \ref{appendix a}. 
\end{remark}

\newcommand{\C}{c}

\algrenewcommand\algorithmicrequire{\textbf{Input:}}
\algrenewcommand\algorithmicensure{\textbf{Output:}}

\section{A PTAS for TBT-SN}

Let $T$ be an instance of TBT-SN, where $T$ has root $r$, and let $\epsilon>0$. We describe our algorithm, which shows Theorem \ref{thmptas}.
Our algorithm consists of the following steps: Cutting the tree $T$ into large subtrees, solving those smaller instances, modifying the solutions of those smaller instances, and combining them to receive a solution for the original instance. Each of the following four subsections deals with one of those steps.

The algorithm uses the parameters $\C$ and $k$,
which only depend on the choice of $\epsilon$. How those parameters are chosen is explained in Section 7.

\subsection{Step 1: Cutting $T$ into pieces of size at least $\C$}

Let $T'$ be a copy of $T$. In the first step of the algorithm, we cut the tree $T'$ into large pieces, resulting in a forest. Iteratively, choose $C$ as the component containing the root $r$, choose a vertex $v$ of biggest depth such that $n(C_v)>\C$ and remove the edge between $v$ and its parent in $T'$. 
The algorithm below gives the details.

\begin{algorithm}[H]
    \caption{Cutting}
    \label{euclid}
    \begin{algorithmic}[1]
    \Require tree $T$
    \State let $T'$ be a copy of $T$ with root $r$
    \State $C:=T'$
        \While{$n(C)>\C$}
\State let $v$ be a vertex of $C$ of greatest depth such that 
$n(C_v)> \C$
\If{$v=r$}
\State \textbf{terminate}
\EndIf
\State remove the edge between $v$ and its parent in $T'$, root the cut-off component in $v$
\State update $C$ to be the component of $T'$ containing $r$
\EndWhile
    \end{algorithmic}
\end{algorithm}

In the following, let $R$ be the component of $T'$ which contains $r$.
Note that every component except 
$R$ has size greater than $\C$.

\subsection{Step 2: Solving the parts}

For every component $C$ with root $v$ of $T'$, we (approximately) solve the instance $C$.
The idea is the following: If $d_C(v)\le k$, then $C$ has at most $\C k+1$ many vertices and we compute a solution of $C$ with brute force. If $d_C(v)>k$, then for an optimal solution $S^*$ of $C$, $\sum_{u\in N_C(v)} \operatorname{dist}_{S^*}(u,v)$,
the cost of the edges adjacent to $v$ in $C$,
dominates the combined cost of all other edges, that is $\sum_{xy\in E(C-v)} \operatorname{dist}_{S^*}(x,y)$. In this case, we compute a solution $S^C$ of $C$ 
such that $\sum_{u\in N_C(v)} \operatorname{dist}_{S^C}(u,v)$ is close to the lower bound $\OPT(K_{1,d_C(v)})$ of $\OPT(C)$ and $\operatorname{cost}(S^C-v)$ is negligible compared to $\OPT(K_{1,d_C(v)})$.

If $C\not=R$, then in both cases we have the further restriction that we want a solution $S^C$ of $C$ where the root $v$ has degree 1.\\
The algorithm below gives the details. For every component $C$ of $T'$, it computes a solution $S^C$.

\begin{algorithm}[H]
    \caption{Solving the components}
    \label{euclid}
    \begin{algorithmic}[1]
    \ForAll{components $C$ of $T'$}
    \State let $v$ be the root of $C$
    \If{$d_C(v)\le k$}
        \If{$C\not=R$}
        \State using brute force, we compute an optimal solution $S^C$ of the instance $C$ under the  
        \Statex \hspace{\algorithmicindent}\hspace{\algorithmicindent}\hspace{\algorithmicindent}additional condition that $v$ has degree 1 in $S^C$, i.e. we compute a binary tree $S^C$ 
\Statex\hspace{\algorithmicindent}\hspace{\algorithmicindent}\hspace{\algorithmicindent}with $V(C)\subseteq V(S^C)$ and $d_{S^C}(v)=1$ that minimizes $\sum_{uv \in E(C)} \operatorname{dist}_{S^C}(u,v)$
        \Else
        \State using brute force, we compute an optimal solution $S^R$ of the instance $C=R$
        \EndIf
    \Else
\State let $u_1,\ldots, u_\ell$ be the neighbors of $v$ in $C$
\State for $i\in [\ell]$, let $C_i$ be the component of $C-v$ that contains $u_i$; we root $C_i$ in $u_i$
\State for $i\in [\ell]$, let $S^{C_i}$ be any binary tree with $V(S^{C_i})=V(C_i)$ and $d_{S^{C_i}}(u_i)\le 1$
\State for $i\in [\ell]$, let $w_i$ be the neighbor of $u_i$ in $S^{C_i}$ (if it exists)
\State let $S^C$ be the binary tree with $V(S^C)=\{v,u_1,\ldots, u_\ell  \}$ and root $v$, such that $u_1$ is the 
\Statex\hspace{\algorithmicindent}\hspace{\algorithmicindent}only child of $v$ (if it exists) and $u_{2i}, u_{2i+1}$ are the children of $u_i$ (if they exist)

\State for $i\in [\ell]$, add the components $S^{C_i}-u_i$ to $S^C$
\State
add edges between the $w_i$'s and the $u_i$'s in $S_C$
such that $u_i$ is an ancestor of $w_i$
and
\Statex\hspace{\algorithmicindent}\hspace{\algorithmicindent}the resulting tree is binary
    \EndIf
    \EndFor
    \end{algorithmic}
\end{algorithm}

\begin{remark}
Note that given an instance $C$, there always exists an optimal solution where every Steiner vertex has degree 3. This implies the existence of an optimal solution where 
less than half of the vertices are Steiner vertices and allows us to find optimal solutions of instances of bounded size with brute force. 
In the following, less than half of the vertices of $S^C$ are Steiner vertices, implying $n(S^C)\le 2 n(C)-1$.

Regarding line 11, in the analysis we only use the bound $\operatorname{cost}(S^{C_i})\le c^2$, hence any solution of $C_i$ with $V(S^{C_i})=V(C_i)$ and $d_{S^{C_i}}(u_i)\le 1$ suffices.\\
Regarding line 15, we already explained in Algorithm~\ref{algorithm_trees_to_binary_trees} line 7 to line 9 how we can add those edges between the $w_i$'s
and the $u_i$'s.
\end{remark}

\subsection{Step 3: Create close attachment points}

We introduce some terminology.
Given a rooted binary tree, we say a vertex of out-degree 0, 1 or 2 has 2, 1 or 0 attachment points, respectively. That is, the number of attachment points of a vertex describes how many edges we can still connect to it. 
We say that an attachment point of a vertex $u$ is below a vertex $v$ if $u$ is a descendant of $v$. The distance between a vertex $v$ and an attachment point of a vertex $u$ is the distance between $v$ and $u$.\\

In the last step, for every component $C$ of $T'$, we found a near-optimal solution $S^C$ where the root has out-degree at most 1 if $C\not=R$.
In this section, we transform each solution $S^C$ of $C$ into a solution $S'^C$ such that for every vertex of $S'^C$, we find a close unique attachment point, which we can then use in the next step to connect the solutions together.

The idea is the following:
Given a solution $S^C$, we consider every vertex $v$ (other than the root) whose (original) distance to the root is a multiple of $\lceil \log_2(\C) \rceil$. If $v$ has no descendant within distance smaller than $\lceil \log_2(\C) \rceil$ of out-degree at most 1, we subdivide the edge between $v$ and its parent by adding a Steiner vertex, which creates a close attachment point for some near ancestors of $v$. Otherwise near ancestors of $v$ already have a close attachment point and we skip $v$.\\
The algorithm below gives the details.
Note that we apply the algorithm to a copy of $S^C$ for every component $C$ of $T'$.


\begin{algorithm}[H]\caption{Create close attachment points}\label{algospace}
    \begin{algorithmic}[1]
    \Require binary tree $S'$ with root $a$
    \State mark every vertex in $S'$ whose distance to $a$ is a multiple of $\lceil \log_2(\C) \rceil$ except $a$ itself
\ForAll{marked vertices $v$}
\If{in $S'_v$ there exists a vertex with out-degree at most 1 with distance less than $\lceil \log_2(\C)\rceil$ to $v$}
    \State unmark $v$
\Else 
\State subdivide the edge between $v$ and its parent $w$, that is,
\Statex \hspace{\algorithmicindent}\hspace{\algorithmicindent}add a new Steiner vertex $z$,
remove the edge $vw$, and add the edges $vz$ and $zw$
%
\State unmark $v$
\EndIf
\EndFor
    \end{algorithmic}
\end{algorithm}

The resulting tree when applying Algorithm~\ref{algospace} to a copy of $S^C$ is called $S'^C$ in the following.\\

Note that Algorithm~\ref{algospace} is designed in such a way that for every vertex $u$ of $S'^C$ there exists a descendant with out-degree at most 1 with distance at most $2\lceil \log_2 \C \rceil$ to $u$.
The following lemma shows that, for vertex $v$ of $S'^C$, there exists a unique attachment point with distance at most $2\lceil \log_2 \C \rceil+1$ to $v$, which is a descendant of $v$. Furthermore, the proof of the lemma shows how a suitable mapping between the vertices of $S'^C$ and its attachment points can be found.

Note that for every component $C$ of $T'$
with at least two vertices, 
the attachment point of the root of $S'^C$ is not used. 
Remember that every component, which does not contain the root of $T'$,
has at least $\C$ vertices.
If the component $R$ of $T'$, which contains the root $r$ of $T$, consists only of the vertex $r$, then we use one attachment point of $r$ for $r$.

\begin{lemma}\label{marble}
Let $F$ be a rooted binary tree 
with at least two vertices.
Let $t\in \mathbb{N}$ such that for every vertex $v$ of $F$ there exists an attachment point, which is a descendant of $v$, with distance at most $t$.
For every vertex $v$ of $F$, there exists a unique attachment point with distance at most $t+1$, which is a descendant of $v$. Furthermore, the attachment point of the root is not used.
\end{lemma}

\begin{proof}
We prove this result algorithmically. Imagine the following process: Place a marble at every vertex of $F$. 
Once a marble is moved to an attachment point, it no longer exists. Our goal is to find a unique attachment point for every marble which has distance at most $t+1$ to its original position. \\
%
After moving the marble from the root to an arbitrary child, 
we iterate the following process.
Consider the vertex $v$ of smallest depth which contains marbles. If there are two marbles at $v$, we place the marble which wasn't originally there, in an attachment point of $v$ if possible or move it to the child where the next attachment point below is closer (break ties arbitrarily). The second marble we either place at an attachment point of $v$, or move it to the child where the next attachment point below is farther away.
If there is only one marble at the vertex $v$, we place the marble at an attachment point of $v$ if possible or move the marble to the child where the next attachment point below is closer.\\

Note, first, that at any moment at most two marbles are at the same vertex.
If two marbles are at a vertex $v$, then, at that moment, all ancestors of $v$ contain no marbles and every descendant of $v$ contains exactly one marble.
Second, a simple induction shows that a rooted tree with $n$ vertices, where all vertices have at most two children, has exactly $n+1$ attachment points. (A rooted tree containing one vertex has two attachment points. In the induction step, adding another vertex removes one attachment point but creates two new ones.)
Hence the binary tree contains enough attachment points and every marble ends up at a unique attachment point.
Third, after a marble moved the first time from its original vertex to a child $u$, it then moves step by step closer to one of the closest attachment points below $u$,
which implies that the marble finds an attachment point with distance at most $1+t$.
\end{proof}

\begin{remark}
    Algorithm~\ref{algospace} can be modified such that no new Steiner vertices are added, that is, we just shift edges.
    Therefore, choose $v$ in Line 2 as the marked vertex of greatest depth and replace Line 6 with
    \begin{algorithmic}[1]
    \State let $u$ be a closest vertex to $v$ of out-degree at most 1 in $S'_v$ 
    \State let $v=v_1v_2\ldots v_\ell=u$ be the path between $v$ and $u$
    \State let $w_{i+1}$ be the second child besides $v_{i+1}$ of $v_i$ for every $i=1,\ldots, \ell-1$.
    \ForAll{$i=1,\ldots, \ell-1$}
    \State remove the edge $v_iw_{i+1}$ and add the edge $v_{i+1}w_{i+1}$
    \EndFor
        \end{algorithmic}
Those changes of course also change the analysis in Section 7. However, after modifying the choice of $\C$ accordingly, we still receive a guaranteed bound which is good enough.
\end{remark}

\subsection{Step 4: Combining the parts}

For every component $C$ of $T'$, we computed a solution $S'^C$ where the root has degree at most 1 if $C\not=R$. For every component $C$ of $T'$ and every vertex $v$ of $S'^C$, the proof of Lemma \ref{marble} shows how to assign a unique attachment point with distance at most $2\lceil \log_2 \C \rceil +1$ to $v$. In this section, we denote the vertex which hosts the attachment point of $v$ by $f(v)$. Again note that, if the component $R$ of $T'$, which contains the root $r$ of $T$, consists only of the vertex $r$, then $f(r)=r$. In all other cases $f(v)$ is not a root.

Our idea is the following. Consider a vertex $u$ of a component $C$ of $T'$, where we removed adjacent edges when we cut the tree $T$ into large components. The set $N^+_T(u) \setminus N^+_{T'}(u)$ contains those cut-off neighbors, which we call $u_1,\ldots, u_\ell$. Let $C_i$ be the component of $T'$ which contains $u_i$. Note that $u_i$ is the root of $C_i$. Furthermore, $u_i$ has degree 1 in $S'^{C_i}$. Let $w_i$ be the child of $u_i$ in $S'^{C_i}$.
We attach the components $S'^{C_1},\ldots, S'^{C_\ell}$ at $f(u)$ using our standard attach-procedure, that is, first remove the $u_i$'s, then build a complete binary tree with the $u_i$'s, add an edge between the root of that binary tree and $f(u)$, and add an edge between $w_i$ and a suitable  descendant of $u_i$ for every $i\in [\ell]$.
The following algorithm formalizes this idea in a bottom-up approach (see Figure \ref{figuret} for illustration). 

\begin{algorithm}[H]
    \caption{Combining the parts}
    \label{euclid}
    \begin{algorithmic}[1]
    \Ensure a solution of the instance $T$    
    \While{not all components of $T'$ are already processed}
    \State 
    choose a previously unchosen component $C$ of $T'$ whose root has maximum depth
    \State let $S$ be a copy of $S'^C$
    \ForAll{$u\in V(C)$}
    \State let $u_1,\ldots, u_\ell$ be the vertices in $N^+_T(u)\setminus N^+_{T'}(u)$
    \State for $i\in [\ell]$, let $C_i$ be the component of $T'$ which contains $u_i$
    \State for $i\in [\ell]$, let $d_i=|N_T(u_i) \cap V(S^{u_i})|$, and w.l.o.g. let $d_1\le d_2 \le \ldots \le d_\ell$

    \State for $i\in [\ell]$, let $w_i$ be the child of $u_i$ in $S^{u_i}$

    \State for $i\in [\ell]$, add $S^{u_i}-u_i$ to $S$, add the vertices $u_1,\ldots, u_\ell$, add an edge between $u_1$ and
    \Statex \hspace{\algorithmicindent}\hspace{\algorithmicindent}$f(u)$ and add edges such that the children of $u_i$ are $u_{2i}$ and $u_{2i+1}$ should they exist
    \State 
add $\ell$ edges between the $w_i$'s and the $u_i's$, which have out-degree at most 1, in such a 
\Statex \hspace{\algorithmicindent}\hspace{\algorithmicindent}way that $u_i$ is an ancestor of $w_i$ and the tree is still binary


\EndFor

\State $S^a:=S$ (if $a$ is the root of $S'^C$)

\EndWhile

\State \textbf{return} $S^r$
    \end{algorithmic}
\end{algorithm}

By construction, for every component $C$ with root $a$, $S^a$ is a solution of $T_a$. In particular, $S^r$ is a solution for $T_r=T$.

All algorithms in Section 5 have polynomial running time. The analysis of the approximation factor follows in Section 7.

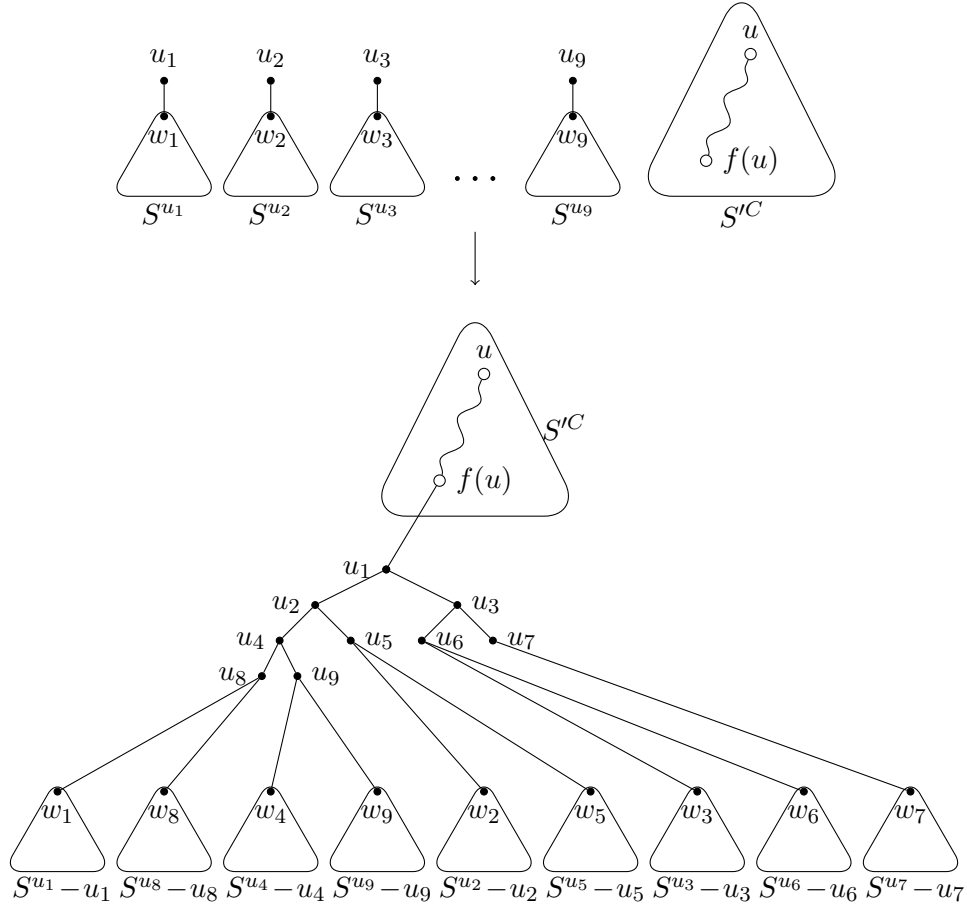
\begin{figure}[H]
    \centering
    \begin{tikzpicture} [scale=0.235] 
    \tikzstyle{point}=[draw,circle,inner sep=0.cm, minimum size=0.9mm, fill=black] 
    \tikzstyle{point2}=[draw,circle,inner sep=0.cm, minimum size=0.5mm, fill=black]
    \tikzstyle{point3}=[draw,circle,inner sep=0.cm, minimum size=1.5mm]
    \tikzstyle{line1}=[line width=0.4mm]

    \begin{scope}[shift={(15,12)}]  
    \draw [rounded corners, rounded corners=5mm] (0,0)--(12,0)--(6,12)--cycle;    
    \begin{scope}[shift={(-18,-12)}]
    \node at (29,17) [label= right:] {$S'^C$};
    
    \node[point3] (fu) at (22,14) [label= right:\text{$f(u)$}] {};
   
    \node[point3] (fm) at (24.5,20) [label=above :\text{$u$}] {};

    \draw[decorate, decoration={snake, amplitude=1mm, segment length=6mm}]
    (fu) -- (fm);
    \end{scope}

    \end{scope}

  \begin{scope}
    \node[point] (u) at (16,9) [label=left:] {};
    
    \draw (fu)--(u);

\begin{scope}[shift={(16,9)}] 
    \node[point] (u1)   at (0,0)  [label=left:$u_1$] {};
\node[point] (u2)  at (-4,-2) [label=left:$u_2$] {};
\node[point] (u3)  at (4,-2)  [label=right:$u_3$] {};

\node[point] (u4)  at (-6,-4) [label=left:$u_4$] {};
\node[point] (u5)  at (-2,-4) [label=right:$u_5$] {};
\node[point] (u6)  at (2,-4)  [label=right:$u_6$] {};
\node[point] (u7)  at (6,-4)  [label=right:$u_7$] {};

\node[point] (u8)  at (-7,-6) [label=left:$u_8$] {};
\node[point] (u9)  at (-5,-6) [label=right:$u_9$] {};

\end{scope}

\draw (u2)--(u1)--(u3);
\draw (u4)--(u2) --(u5);
\draw (u6)--(u3)--(u7);
\draw (u8)--(u4)--(u9);
\end{scope}

\foreach \i/\t/\x in {
    1/u8/0,
    8/u8/1,
    4/u9/2,
    9/u9/3,
    2/u5/4,
    5/u5/5,
    3/u6/6,
    6/u6/7,
    7/u7/8
} {

\begin{scope}[shift={({6*\x-5},-7)}]

    \node[point] (v\i) at (2.5,3.5) [label=below:$w_{\i}$] {};
\node at (2.5,-2)
[label=below:] {{ $S^{u_\i}\!-\!u_\i$}};
    \draw[rounded corners=2.5mm]
        (-0.56,-1)--(2.5,4.3)--(5.56,-1)--cycle;

    \draw (v\i) -- (\t);

\end{scope}
}

\begin{scope}[shift={(21,28)}]    
\draw[->] (0,0) -- (0,-3);
\end{scope}

\begin{scope}[shift={(0,51)}]

\foreach \i in {1,2,3} {
\begin{scope}[shift={({6*\i-5},-20)}]

    \node[point] (v\i) at (2.5,3.5) [label=below:$w_{\i}$] {};
    \node[point] (g\i) at (2.5,5.5) [label=above:$u_{\i}$] {};
    \draw (v\i)--(g\i);
\node at (2.5,-2)
[label=below:] {{$S^{u_\i}$}};
    \draw[rounded corners=2.5mm]
        (-0.56,-1)--(2.5,4.3)--(5.56,-1)--cycle;
\end{scope}
}

    \node[point2] at (20,-20) [label=below:] {};
    \node[point2] at (21,-20) [label=below:] {};
    \node[point2] at (22,-20) [label=below:] {};

\foreach \i in {9} {
\begin{scope}[shift={(24,-20)}]

    \node[point] (v\i) at (2.5,3.5) [label=below:$w_{\i}$] {};
    \node[point] (g\i) at (2.5,5.5) [label=above:$u_{\i}$] {};
    \draw (v\i)--(g\i);
\node at (2.5,-2)
[label=below:] {{$S^{u_\i}$}};
    \draw[rounded corners=2.5mm]
        (-0.56,-1)--(2.5,4.3)--(5.56,-1)--cycle;
\end{scope}
}

 \begin{scope}[shift={(30,-21)}]  
    \draw [rounded corners, rounded corners=5mm] (0,0)--(12,0)--(6,12)--cycle;    
    \node at (6,-1) [label= right:] {$S'^C$};
     \begin{scope}[shift={(-18,-12)}]
    \node[point3] (fu) at (22,14) [label= right:\text{$f(u)$}] {};
   
    \node[point3] (fm) at (24.5,20) [label=above :\text{$u$}] {};

    \draw[decorate, decoration={snake, amplitude=1mm, segment length=6mm}]
    (fu) -- (fm);
    \end{scope}
    \end{scope}

\end{scope}

\end{tikzpicture}

\caption{Illustration for $\ell=9$, that is, $\{u_1,\ldots, u_\ell\}=N^+_T(u)\setminus N^+_{T'}(u)$. $S^{u_1},\ldots, S^{u_9}$ get attached at attachment point $f(u)$ of $S'^C$. }\label{figuret}

\end{figure}

\section{Analysis of the $\frac{8}{3}$-approximation algorithm for general instances}

In the first part of this section, we prove the claimed approximation ratio of $\frac{8}{3}$ of Algorithm~\ref{algorithm_trees_to_binary_trees} for general instances $T$. In the second part, we show that the cost of the tree $S$ created by Algorithm \ref{algorithm_trees_to_binary_trees} becomes arbitrarily close to the optimal solution if every vertex in $T$ is either a leaf or has a sufficiently large degree. 

\subsection{Worst Case Analysis for general trees} \label{section_worst_case_analysis}

Let $T$ be an arbitrary tree and $S$ be the tree created by Algorithm~\ref{algorithm_trees_to_binary_trees} with input $T$. Consider an arbitrary vertex $v \in V(T)$ with children $u_1, \ldots, u_d$. We split our analysis as well as the overall distance between $v$ and its children $u_1, \ldots, u_d$ in $S$ into two parts. To this end, consider the iteration in which Algorithm~\ref{algorithm_trees_to_binary_trees} builds the tree $S^v$. Firstly, we consider the overall distance between $v$ and all its children $u_1, \ldots, u_d$ in $S^v_+$ which is
\[
\operatorname{cost}(S^v_+) = \sum_{i = 1}^{d} \operatorname{dist}_{S^v_+}(u_i, v).
\] 
Secondly, we consider the additional distance between the vertices $u_1, \ldots, u_d$ and their respective children, which may be caused by deleting the direct edge between $u_i$ and $w_i$ in line \ref{algorithm_8over3_start_creating_Sv+} and reattaching $S^{u_i}-u_i$ to some leaf in $S^{v}_+(u_i)$ in line \ref{algorithm_8over3_connecting_subtrees} of Algorithm~\ref{algorithm_trees_to_binary_trees}. These additional distances are then caused if the leaf to which $S^{u_i}- u_i$ is attached is not $u_i$ but some other child of $v$.\\
We lower bound the overall distance of $u_1, \ldots, u_d$ to their parent $v$ in any binary tree $G$ with $V(G) = V(T)$. To this end, observe that due to the constraint on the degree in $G$ there can be at most $3$ children $u_i$ of $v$ with $\operatorname{dist}_G(u_i,v) = 1$, at most another $6$ children $u_i$ with $\operatorname{dist}_G(u_i,v) = 2$ and so on. Hence, the overall minimal distance between $v$ and its children in any binary tree $G$ with $V(G) = V(T)$ is achieved if $u_1, \ldots, u_d$ form three complete binary trees whose heights differ by at most $1$ and whose roots are adjacent to $v$.\\
To ease notation when considering the lower bound on the overall distance cost between $v$ and its children, we assume w.l.o.g.\ that these three complete binary trees are rooted in $u_1$, $u_2$ and $u_3$ respectively and that the children of $u_i$ are $u_{2i+2}$ and $u_{2i+3}$, if the respective vertices exist. We define $G^v$  as the tree with these properties for $v$ and its children $u_1, \ldots, u_d$. A depiction of the respective subtrees can be seen in Figure \ref{fig_lower_bound_opt_and_Sv}.\\ 
Note that the order of the children $u_1, \ldots, u_d$ around $v$ is not relevant for our lower bound. The only subject of importance is that the children $u_1, \ldots, u_d$ are as close to their parent $v$ as possible which is the case in $G^v$. Hence, the tree $G^v$ provides a lower bound for the overall distance between the vertex $v$ and its children $u_1, \ldots, u_d$ w.r.t.\ any other binary tree $G$ with $V(G) = V(T)$.\\
To be more precise, we obtain $\operatorname{dist}_{S^v_+}(u_i, v) = k$ for all children $u_i$ of $v$ with indices $i = 2^{k-1}, \ldots, 2^k - 1$. Meanwhile, we can assume for any binary tree $G$ with $V(G) = V(T)$ w.l.o.g.\ that  $\operatorname{dist}_G(u_i, v) \geq k$  for all children $u_i$ of $v$ with indices $i = 3 \cdot (2^{k-1} - 1) + 1, \ldots, 3 \cdot (2^k - 1) = 3 \cdot 2^{k-1} - 2, \ldots, 3 \cdot 2^k - 3$. We formalize these observations in the following lemma.

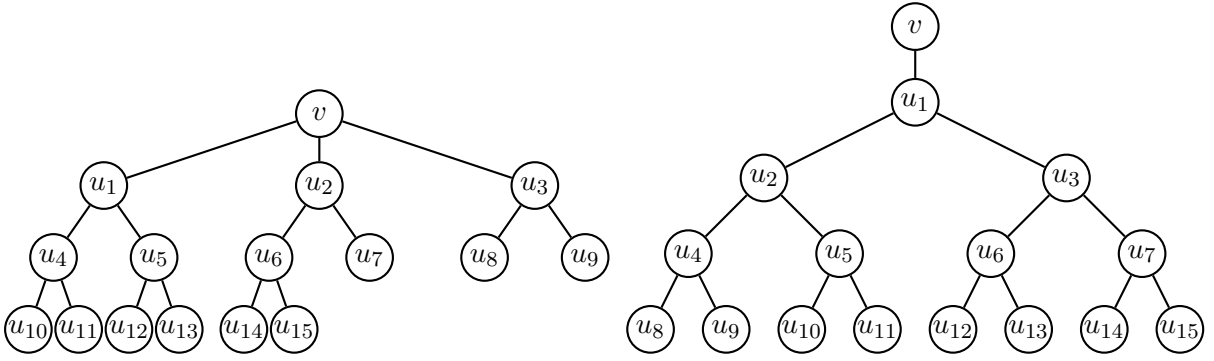
\begin{figure} 
    \centering
    \begin{tikzpicture}[thick, level distance = 10mm, level 1/.style={sibling distance = 30mm}, level 2/.style={sibling distance = 14mm}, level 3/.style={sibling distance = 7mm}, every node/.style={circle, draw, inner sep=0, minimum size=1.6em}, scale=0.95]
    \node (root) {$v$}
        child {node {$u_1$} child {node {$u_4$} child {node {$u_{10}$}}  child {node {$u_{11}$}} }  child {node {$u_5$} child {node {$u_{12}$}}  child {node {$u_{13}$}}}}
        child {node {$u_2$} child {node {$u_6$} child {node {$u_{14}$}} child {node {$u_{15}$}}} child {node {$u_7$}}}
        child {node {$u_3$} child {node {$u_8$}}  child {node {$u_9$}}};
    \end{tikzpicture} \hfill
    \begin{tikzpicture}[thick, level distance = 10mm, level 1/.style={sibling distance = 30mm}, level 2/.style={sibling distance = 40mm}, level 3/.style={sibling distance = 20mm}, level 4/.style={sibling distance = 10mm}, every node/.style={circle, draw, inner sep=0, minimum size=1.6em}]
    \node (root) {$v$}
        child {node {$u_1$} 
            child {node {$u_2$} 
                child {node (u4) {$u_{4}$} child {node {$u_{8}$}} child {node {$u_{9}$}}}  
                child {node {$u_{5}$} child {node {$u_{10}$}} child {node {$u_{11}$}}}}  
            child {node {$u_3$} 
                child {node {$u_{6}$} child {node {$u_{12}$}} child {node {$u_{13}$}}}  
                child {node {$u_{7}$} child {node {$u_{14}$}} child {node {$u_{15}$}}}}}
        ;
\end{tikzpicture}
    \caption{\centering Left: The tree $G^v$ described in Section~\ref{section_worst_case_analysis} to bound the overall distance of $v$ to its children; Right: The tree $S^v_+$ created in lines \ref{algorithm_8over3_start_creating_Sv+}-\ref{algorithm_8over3_end_creating_Sv+} of Algorithm~\ref{algorithm_trees_to_binary_trees}.}
    \label{fig_lower_bound_opt_and_Sv}
\end{figure}

\begin{lemma}\label{lemma__cost_building_subtrees}
    Let the trees $G^v$ be defined as above for all $v \in V(T)$. We define
    \[
    \operatorname{cost}(S^v_+)(j) \coloneqq \sum_{i = 1}^{j} \operatorname{dist}_{S^v_+}(u_i, v) \text{\qquad and \qquad} \operatorname{cost}(G^v)(j) \coloneqq \sum_{i = 1}^{j} \operatorname{dist}_{G^v}(u_i, v)
    \]
    as the overall distance between $v$ and its first $j$ children in $S^v_+$ and $G^v$ respectively. Then
    \[
    \frac{\operatorname{cost}(S^v_+)(j)}{\operatorname{cost}(G^v)(j)} \leq \frac{5}{3}
    \]
    for all $j \geq 1$.
\end{lemma}
\begin{proof}
    For $j=1, \ldots, 9$, the respective ratios can be seen in Table \ref{table_case_1_to_10}. Reconsider our analysis above and observe that for $k \geq 3$, we obtain $\operatorname{dist}_{S^v_+}(u_i, v) = k+1$ and $\operatorname{dist}_{G^v}(u_i,v) = k-1$ for all $u_i$ with $i = 2^{k}, \ldots, 3 \cdot 2^{k-1} - 3$. Similarly, we obtain $\operatorname{dist}_{S^v_+}(u_i, v) = k+1$ and $\operatorname{dist}_{G^v}(u_i,v) = k$ for all $u_i$ with $i = 3 \cdot 2^{k-1}-2, \ldots, 2^{k+1}-1$.\\
    Let $j \geq 9$ and observe the increase in $\operatorname{cost}(S^v_+)(j)$ and $\operatorname{cost}(G^v)(j)$ respectively when considering an additional vertex $u_{j+1}$. If $j+1 \leq 3 \cdot 2^{\floor{\log_2(j+1)} - 1} - 3$, then $\operatorname{cost}(S^v_+)(j)$ increases by $\floor{\log_2(j+1)} + 1$ and $\operatorname{cost}(G^v)(j)$ increases by $\floor{\log_2(j+1)} - 1$. 
    The conditions $j \geq 9$ and
$j+1 \leq 3 \cdot 2^{\floor{\log_2(j+1)}-1}-3$
imply that $j \geq 15$. Hence,
    \[
    \frac{\floor{\log_2(j+1)} + 1}{\floor{\log_2(j+1)} - 1} \leq \frac{5}{3}
    \]
    by monotonicity.\\
    If $3 \cdot 2^{\floor{\log_2(j+1)} - 1} - 2 \leq  j+1 \leq 2^{\floor{\log_2(j+1)}+1} - 1$, then $\operatorname{cost}(S^v_+)(j)$ increases by $\floor{\log_2(j+1)} + 1$ and $\operatorname{cost}(G^v)(j)$ increases by $\floor{\log_2(j+1)}$. Hence, 
    \[
    \frac{\floor{\log_2(j+1)} + 1}{\floor{\log_2(j+1)}} \leq \frac{4}{3} \leq  \frac{5}{3}
    \]
    by monotonicity.\\
    Together with the induction hypothesis, we can conclude that
    \[
    \frac{\operatorname{cost}(S^v_+)(j+1)}{\operatorname{cost}(G^v)(j+1)} \leq \frac{5}{3}
    \]
    in both cases, proving our claim.
\end{proof}

\begin{table}[ht]
    \centering
    \begin{tabular}{|c|c|c|c|} 
        \hline
        $j$ & $\operatorname{cost}(S^v_+)(j)$ & $\operatorname{cost}(G^v)(j)$ & $\frac{\operatorname{cost}(S^v_+)(j)}{\operatorname{cost}(G^v)(j)}$  \\[6pt]
        \hline
        $1$ & $1$ & $1$ & $1$ \\[3pt]
        $2$ & $3$ & $2$ & $\frac{3}{2}$ \\[3pt]
        $3$ & $5$ & $3$ & $\frac{5}{3}$ \\[3pt]
        $4$ & $8$ & $5$ & $\frac{8}{5}$ \\[3pt]
        $5$ & $11$ & $7$ & $\frac{11}{7}$ \\[3pt]
        $6$ & $14$ & $9$ & $\frac{14}{9}$ \\[3pt]
        $7$ & $17$ & $11$ & $\frac{17}{11}$ \\[3pt]
        $8$ & $21$ & $13$ & $\frac{21}{13}$ \\[3pt]
        $9$ & $25$ & $15$ & $\frac{5}{3}$ \\[3pt]
        \hline
    \end{tabular}\\
    \caption{Sum of distances between the first $j$ children and the root $v$ in $S^v_+$ for $j=1, \ldots, 9$ compared to the overall distances in $G^v$.}
    \label{table_case_1_to_10}
\end{table}

\begin{corollary} \label{corollary_bound_for_arranging_children}
    Let $\sum_{v \in V(T)} \operatorname{cost}(S^v_+)$ be the sum of the distance costs for all vertices $v \in V(T)$ due to the subtrees created in line \ref{algorithm_8over3_end_creating_Sv+} of Algorithm~\ref{algorithm_trees_to_binary_trees}. Then
    \[
    \frac{\sum_{v \in V(T)} \operatorname{cost}(S^v_+)}{\OPT(T)} \leq \frac{5}{3}
    \]
\end{corollary}
\begin{proof}
    Follows immediately from Lemma~\ref{lemma__cost_building_subtrees} when bounding $\OPT(T)$ locally w.r.t.\ $G^v$ and summing up over all vertices $v \in V(T)$. 
\end{proof}

\noindent Again, let $v \in V(T)$ be some arbitrary vertex and fix the iteration in which Algorithm~\ref{algorithm_trees_to_binary_trees} considers $v$. We bound the additional costs that we potentially incur by not having a direct edge between $w_i$ and $u_i$ anymore after adding the subtrees $S^{u_i}-u_i$ to $S^v_+$ in line \ref{algorithm_8over3_connecting_subtrees} of Algorithm~\ref{algorithm_trees_to_binary_trees}.\\ 
To this end, let $B$ be some arbitrary (binary) rooted tree. We define the level sets of $B$ 
\[
L_k(B) \coloneqq \setbuild{v \in V(B)}{\max_{u \in B_v} \operatorname{dist}_{B_v}(u,v) = k}
\]
which contain all vertices $v \in V(B)$ which have a distance of $k$ to the leaf furthest away in their respective subtree $B_v$.\\
Let $k \in \{1, \ldots, \floor{\log_2(d)}\}$. Then, for every vertex $u_i \in L_k(S^v_+)$ and all its children $w_{i,j} \in N^+_T(u_i)$ the distance $\operatorname{dist}_{S^v}(u_i, w_{i,j})$ between $u_i$ and the $w_{i,j}$ increases by at most $k$ compared to the distance $\operatorname{dist}_{S^{u_i}}(u_i, w_{i,j})$ in the subtree $S^{u_i}$. Note that $\operatorname{dist}_{S^v}(u_i, w_{i,j}) = \operatorname{dist}_{S}(u_i, w_{i,j})$, i.e.\ the distance between $u_i$ and its children does not change after this iteration. We further define in the following $l_k(B) \coloneqq \abs{L_k(B)}$ for any binary tree $B$.
\begin{remark}
    Note that $L_k(S^v_+) = \emptyset$ for all $k \geq \floor{\log_2(d)} +2$ and $L_k(S^v_+) = \{v\}$ for $k = \floor{\log_2(d)} +1$, since the children of $v$ form a complete binary tree whose root is connected to $v$ in $S^v_+$.
\end{remark}

\begin{lemma} \label{lemma_subtree_adding_cost}
    Let $B$ be some complete binary tree with $n$ vertices.
    Then 
    \[
    \sum_{k= 1}^{\floor{\log_2(n)}} k \cdot l_k(B) \leq n-1.
    \]
\end{lemma}
\begin{proof}
    Induction w.r.t.\ the height $h= \floor{\log_2(n)}$ of $B$. For $h=0$ the tree $B$ only has one vertex and the claim holds. Let $h \geq 0$ and assume that the claim holds for every complete binary tree $B'$ with $n'$ vertices such that $\floor{\log_2(n')} = h$.\\
    Consider some complete binary tree $B$ with $n$ vertices such that $\floor{\log_2(n)} = h+1$. By Fact \ref{fact_complete_binary_tree}, the number of leaves of a complete binary tree with $n$ vertices is at least $\frac{n}{2}$. If we remove all leaves from $B$, we are left with some tree $B'$ with $n' \leq \frac{n}{2}$ vertices and height $\floor{\log_2(n)} - 1$. From our observations and the induction hypothesis it follows that
    \begin{equation} \label{lemma_subtree_adding_cost_eq1}
        \sum_{k= 1}^{\floor{\log_2(n')}} k \cdot l_k(B') = \sum_{k= 1}^{\floor{\log_2(n)} - 1} k \cdot l_k(B') \leq \frac{n}{2}-1.
    \end{equation}
    Further, consider some vertex $u \in V(B')$ and observe that adding the leaves again would increase the distance of $u$ to its furthest child in the subtree $B'(u)$ by exactly $1$. Consequently, $L_k(B') = L_{k+1}(B)$ for all $k = 0, \ldots, \floor{\log_2(n)}-1$ and 
    \[
    \sum_{k= 1}^{\floor{\log_2(n)}} k \cdot l_k(B)  = \sum_{k= 0}^{\floor{\log_2(n)} - 1} (k+1) \cdot l_k(B') \leq \frac{n}{2} - 1 + \frac{n}{2} = n-1.
    \]
    The inequality follows from \eqref{lemma_subtree_adding_cost_eq1} and the fact that $\sum_{k= 0}^{\floor{\log_2(n)} - 1} l_k(B') \leq \frac{n}{2}$, since $B'$ has at most $\frac{n}{2}$ vertices. 
\end{proof}
\begin{remark}
    The bound in Lemma~\ref{lemma_subtree_adding_cost} is tight if $n = 2^k$ for some $k \in \NN$.
\end{remark}

\noindent With Lemma~\ref{lemma_subtree_adding_cost}, we can now bound the additional costs incurred when adding the subtrees $S^{u_i}-u_i$  to $S^v_+$ in line \ref{algorithm_8over3_connecting_subtrees} of Algorithm~\ref{algorithm_trees_to_binary_trees} by the total number of children of the vertices $u_1, \ldots, u_d$.

\begin{lemma}\label{lemma_adding_subtree_cost}
    Let $v \in V(T)$ be some vertex in $T$ with children $u_1, \ldots, u_d$. Then
    \[
    \sum_{i = 1}^{d} \sum_{w \in N_T^+(u_i)} \operatorname{dist}_{S^v}(u_i, w) - \operatorname{dist}_{S^{u_i}}(u_i, w) \leq \sum_{i = 1}^{d} d_T^+(u_i).
    \]
\end{lemma}
\begin{proof}
    For every child $u_i \in L_k(S^v_+)$ of $v$ the distance $\operatorname{dist}_{S^v}(u_i, w)$ to its children $w \in N_T^+(u_i)$ increases by at most $k$ compared to the distance $\operatorname{dist}_{S^{u_i}}(u_i, w)$ in the subtree $S^{u_i}$. Hence, the additional costs incurred by adding the subtrees $S^{u_i}-u_i$ to $S^v_+$ in line \ref{algorithm_8over3_connecting_subtrees} of Algorithm~\ref{algorithm_trees_to_binary_trees} are 
    \begin{equation} \label{equation_proof_adding_subtree_cost}
        \sum_{k = 1}^{\floor{\log_2(d)}} k \sum_{u_i \in L_k(S^v_+)} d_T^+(u_i) \leq \sum_{k = 1}^{\floor{\log_2(d)}} k \sum_{u_i \in L_k(S^v_+)} \frac{1}{d} \sum_{i = 1}^{d} d_T^+(u_i),
    \end{equation}
    due to the construction of $S^v_+$ in lines \ref{algorithm_8over3_start_creating_Sv+}-\ref{algorithm_8over3_end_creating_Sv+} and the requirement $d_T^+(u_1) \leq \ldots \leq d_T^+(u_d)$ in line \ref{algorithm_8over3_choosing_v_and_sorting_children} of Algorithm~\ref{algorithm_trees_to_binary_trees}. Further, note that the distance between the vertices $u_{\frac{d}{2}+1}, \ldots u_d$ and their respective children does not increase. Hence, averaging over all children of $u_1, \ldots, u_d$ can only increase the sum on the l.h.s.\ of \eqref{equation_proof_adding_subtree_cost}. The claim follows by using Lemma~\ref{lemma_subtree_adding_cost}
    \begin{align*}
        \sum_{k = 1}^{\floor{\log_2(d)}} k \sum_{u_i \in L_k(S^v_+)} \frac{1}{d} \sum_{i = 1}^{d} d^+_T(u_i) &= \pare{\frac{1}{d} \sum_{i = 1}^{d} d^+_T(u_i)} \sum_{k = 1}^{\floor{\log_2(d)}} k \cdot l_k(S^v_+)&\\
        &\leq \pare{\frac{1}{d} \sum_{i = 1}^{d} d_T^+(u_i)} \cdot (d-1) & \text{(by Lemma~\ref{lemma_subtree_adding_cost})}\\
        &\leq \sum_{i = 1}^{d} d_T^+(u_i).&
    \end{align*}
\end{proof}

\begin{corollary}\label{corollary_costs_adding_subtrees}
   The overall additional costs for adding the subtrees in line \ref{algorithm_8over3_connecting_subtrees} of Algorithm~\ref{algorithm_trees_to_binary_trees} are at most $\OPT(T)$.
\end{corollary}
\begin{proof}
    Let $v \in V(T)$ be some arbitrary vertex with children $u_1, \ldots, u_d$. By Lemma \ref{lemma_adding_subtree_cost}, the cost of adding all subtrees $S^{u_i}-u_i$ for $i = 1, \ldots, d$ to $S^v_+$ is at most $\sum_{i = 1}^{d} d^+_T(u_i)$. The distance $\operatorname{dist}_{S^*}(u_i, w)$ between $u_i$ and its children $w \in N^+_T(u_i)$ is at least $1$ for any optimal solution $S^*$ with $V(S^*) = V(T)$. Hence, 
    \[
    \sum_{i=1}^{d} \sum_{w \in N_T^+(u_i)} \operatorname{dist}_{S^*}(u_i, w) \geq \sum_{i = 1}^{d} d_T^+(u_i)
    \]
    Summing up over all vertices $v \in V(T)$ yields our claim.
\end{proof}
\noindent Finally, we can sum up our results to prove the approximation ratio of Algorithm~\ref{algorithm_trees_to_binary_trees}.
\begin{corollary}
    Let $T$ be some arbitrary tree and $S$ be the tree created by Algorithm~\ref{algorithm_trees_to_binary_trees} with input $T$. Then 
    \[
    \frac{\operatorname{cost}_T(S)}{\OPT(T)} \leq \frac{8}{3}.
    \]
\end{corollary}
\begin{proof}
    The total cost of Algorithm~\ref{algorithm_trees_to_binary_trees} consists of the costs incurred when creating the subtrees $S^v_+$ in lines \ref{algorithm_8over3_start_creating_Sv+}-\ref{algorithm_8over3_end_creating_Sv+} and the potential increase in costs when adding the subtrees $S^{u_i}-u_i$ to $S^v_+$ in line \ref{algorithm_8over3_connecting_subtrees}. Hence, by Corollaries \ref{corollary_bound_for_arranging_children} and \ref{corollary_costs_adding_subtrees}
    \[
    \frac{\operatorname{cost}_T(S)}{\OPT(T)} \leq \frac{\sum_{v \in V(T)} \operatorname{cost}(S^v_+) + \OPT(T)}{\OPT(T)} \leq \frac{5}{3} + 1 = \frac{8}{3}.
    \]
\end{proof}

\subsection{Convergence for large degrees} \label{section_convergence_for_large_degrees}

Consider a tree $T$ in which every internal vertex has a degree of at least $\Delta + 1$, i.e.\ in which every vertex has either at least $\Delta$ or $0$ children. We show that the approximation ratio of Algorithm~\ref{algorithm_trees_to_binary_trees} tends to $1$ as $\Delta \rightarrow \infty$. 

\begin{lemma}\label{lemmaToOne}
    Let $cost(S^v_+)(j)$ and $\cost(G^v)(j)$ be defined as in Lemma~\ref{lemma__cost_building_subtrees}. Then
    \[
    \frac{\cost(S^v_+)(j)}{\cost(G^v)(j)} \xrightarrow{j \rightarrow \infty} 1.
    \]
\end{lemma}

\begin{proof}
     Consider the first $j$ children $u_1, \ldots, u_j$ of $v$.
     Note that $0\le \operatorname{dist}_{S^v_+}(v,u_i)- \operatorname{dist}_{G^v}(v,u_i) \leq 2$ for every $i =1,\ldots, j$. Hence $0 \le  \cost(S^v_+)(j) - \cost(G^v)(j) \le 2j  $.
According to \cite{paper}, the optimal value for a star graph with $n+1$ vertices satisfies $\OPT(K_{1,n}) = \Omega(n \log(n))$, implying $\cost(G^v)(j) = \Omega(j \log(j))$. Hence, 
     $$ \frac{\cost(S^v_+)(j)}{\cost(G^v)(j)} \le  1+\frac{2j}{\cost(G^v)(j)} \xrightarrow{j \rightarrow \infty} 1.$$
\end{proof}

\begin{lemma}\label{lemma_large_degree_to_one}
    Assume that every internal vertex in $T$ has degree at least $\Delta + 1$. Then 
    \[
    \limsup_{\Delta\to\infty} \frac{\sum_{v \in V(T)} \operatorname{cost}(S^v_+)}{\OPT(T)} \le  1.
    \]
\end{lemma}
\begin{proof}
    Follows immediately from Lemma~\ref{lemmaToOne} and the fact that according to the definition in Lemma~\ref{lemma__cost_building_subtrees},
    \[
    \OPT(T) \geq \sum_{v \in V(T)} \operatorname{cost}(G^v)(d^+_T(v)).
    \]
\end{proof}

\noindent We now take a look at the additional costs incurred when adding the subtrees $S^{u_i}-u_i$ in line \ref{algorithm_8over3_connecting_subtrees} of Algorithm~\ref{algorithm_trees_to_binary_trees}. By Lemma \ref{lemma_adding_subtree_cost} we know that the total increase of the distance between neighboring vertices in $T$ is at most $n(T)$. The following Theorem shows that relative to the value of the optimal solution these costs converge to $0$ as $\Delta \rightarrow \infty$. 
\begin{lemma} \label{lemma_large_degree_to_zero}
    Assume that every internal vertex of $T$ has degree at least $\Delta + 1$. Then
    \[
    \frac{n(T)}{\OPT(T)} \xrightarrow{\Delta \rightarrow \infty} 0.
    \]
\end{lemma}
\begin{proof}
According to \cite{paper}, the optimal value for a star graph with $n+1$ vertices satisfies $\OPT(K_{1,n}) = \Omega(n \log(n))$. Hence, we can bound $\OPT(T) \geq \Omega(\log(\Delta)) n(T)$ and the claim follows immediately.
\end{proof}
\noindent We conclude our claim from the beginning of this section in the following corollary. 
\begin{corollary}
    Assume that every internal vertex of $T$ has degree at least $\Delta + 1$. Let $S$ be the tree created by Algorithm~\ref{algorithm_trees_to_binary_trees} with input $T$. Then
    \[
    \frac{\operatorname{cost}_T(S)}{\OPT(T)} \xrightarrow{\Delta \rightarrow \infty} 1.
    \]
\end{corollary}
\begin{proof}
    By Lemmas \ref{lemma_adding_subtree_cost}, \ref{lemma_large_degree_to_one}, and \ref{lemma_large_degree_to_zero}, we can bound 
    \[
    \limsup_{\Delta\to\infty} \frac{\operatorname{cost}_T(S)}{\OPT(T)} \leq \limsup_{\Delta\to\infty} \frac{\sum_{v \in V(T)} \operatorname{cost}(S^v_+) + n(T) }{\OPT(T)} \le 1.
    \]
    Now $\operatorname{cost}_T(S) \ge \OPT(T) $ implies the result.
\end{proof}

\section{Analysis of the PTAS}

In this section, we show that the algorithm from Section 5 achieves the
$(1+\epsilon)$-approximation guarantee, as stated in Theorem~\ref{thmptas}.
To simplify notation, we allow constant multiples of $\epsilon$, such as
$(1+4\epsilon)\OPT$. These factors can be compensated for by scaling appropriately.\\
As observed by \cite{paper}, for any vertex $v$ of $T$ and any solution $S$ of $T$,
$$ \sum\limits_{u\in N_T(v)} \operatorname{dist}_{S}(u,v) \in \Omega\left(d_T(v)\log d_T(v)\right).$$
Hence there exist constants $\alpha > 0$ and $d_0 \in \mathbb{N}$ such that for every tree $T$ and every vertex $v$ of $T$ with degree at least $d_0$, the combined costs of the edges adjacent to $v$ in any solution $S$ are at least 
$\alpha d_T(v)\log_2 d_T(v)$, that is
$$ \sum\limits_{u\in N_T(v)} \operatorname{dist}_{S}(u,v) \ge \alpha d_T(v)\log_2 d_T(v).$$

We fix $\Delta$ such that $\Delta> d_0$ and
$d < \frac{1}{2} \epsilon \alpha d\log_2 d$ for all $d\ge \Delta$.\\
Lemma \ref{lemmaToOne} implies the existence of $D$ such that for every $\ell \ge D$, the following is true:
Let $S'$ be the binary tree with vertices $v,u_1,\ldots, u_\ell$, such that $u_1$ is the only child of $v$ and $u_{2i}, u_{2i+1}$ are the children of $u_i$ (if they exist). Now
$$\sum\limits_{i=1}^\ell \operatorname{dist}_{S'}(v,u_i) \le (1+\epsilon) \OPT(K_{1,\ell}) .$$
We fix such a constant $D$.\\
Let $p := \left\lceil \frac{2}{\epsilon} \right\rceil + 1$.
The parameter $\C$ is chosen large enough to fulfill the following properties.
\begin{align*}
3 &\le   \epsilon \lceil \log_2(\C) \rceil\\
D&<\epsilon (\C-1)\\
\Delta &< \epsilon (\C-1)\\
2\lceil \log_2(\C) \rceil +1 &\le \epsilon (\C-1)\\
6(2^{p-1}-1)+6(2^{p-1}-1)^2&\le \epsilon (\C-1) \\
\Delta \left(2^{p+1}-1\right)&< \epsilon (\C-1)
\end{align*}

The parameter $k$ is chosen large enough such that $k\ge D$ and 
\begin{align*}
c^2K+\C K &\le \epsilon \alpha K \log_2 K \text{ for all } K\ge k.
\end{align*}

The following two inequalities will be used several times during the analysis.\\
$\OPT\left(K_{1,\ell} \right)$ is the cost of an optimal solution of a star with $\ell+1$ vertices. 
For a tree $T$ and a solution $S$, it clearly holds that
$$ \sum\limits_{u\in N_T(v)} \operatorname{dist}_{S}(u,v) \ge \OPT\left(K_{1,d_T(v)} \right). $$
Summing over all vertices implies 
%
$$\sum\limits_{ v \in V(T)}\OPT\left(K_{1,d^+_{T}(v)}\right)\le \sum\limits_{ v \in V(T)}\OPT\left(K_{1,d_{T}(v)}\right) \le 2\OPT(T)$$
for a rooted tree $T$.

\subsection{Analysis of Algorithm 3}

Let $C$ be a component of $T'$ with root $v$ such that $d_C(v)\le k$. Forcing the vertex $v$ to have degree 1 in the solution $S^C$ changes the value of the optimal solution only slightly, since $C$ has by choice many vertices. This is captured by the following lemma.

\begin{lemma}\label{lemmaA1} Let $C$ be a component of $T'$.
    If $d_v(C)\le k$, then
     $$\operatorname{cost}(S^C)\le (1+2\epsilon)\OPT(C).$$
\end{lemma}


\begin{proof}
Let $S^{*C}$ be an optimal solution of the instance $C$ (without the degree-condition on $v$). 
If $v$ has degree 1 in $S^{*C}$, the claim follows. Now, let $v$ have degree at least 2.

Assume $v$ has degree 3 in $S^{*C}$. Let $v_1,v_2,v_3$ be the neighbors of $v$. We create a new solution $S'$ from $S^{*C}$ as follows: 
We remove the edges $vv_1,vv_2,vv_3$, add two new Steiner vertices $x_1,x_2 $, and add the edges $vx_1$, $v_1x_1$, $x_1v_2$, $x_1x_2$, and $x_2v_3$.
This increases the distances by at most 2.\\

The idea of the proof is the following: 
given $xy\in E(T)$ with $x,y\in V(C)$, if the distance between $x,y$ in $S^{*C}$ is big (at least $p$), then increasing the distance by 2 changes little. On the other hand, there are only finitely many such $xy$ whose distance increases and whose distance is smaller $p$.\\

Let $V_1,V_2,V_3$ be the components of $S^{*C}$ which contain $v_1,v_2,v_3$, respectively.
Note that for every edge $xy\in E(T)$ with $x,y\in V(C)$, the distance between $x$ and $y$ in $S^{*C}$ and in $S'$ is equal if $x,y$ are in the same component of $S^{*C}-v$.
By the choice of $p$, if $\operatorname{dist}_{S^{*C}}(x,y)\ge p$ for some $xy\in E(T)$ with $x,y\in V(C)$, then $\operatorname{dist}_{S'}(x,y)\le \operatorname{dist}_{S^{*C}}(x,y)+2\le (1+\epsilon) \operatorname{dist}_{S^{*C}}(x,y) $.
We bound the number of edges $x,y\in E(T)$ with $xy\in V(C)$ such that $\operatorname{dist}_{S^{*C}}(x,y)<p $ and $x,y$ are not in the same component of $S^{*C}-v$, that is $$E'=\{ xy\in E(T): x,y\in V(C), \operatorname{dist}_{S^{*C}}(x,y) <p, |\{x,y\}\cap V_i|\le 1 \text{ for all }i=1,2,3\}. $$

For every $i\in [3]$, there are at most $2^{p-1}-1$ vertices in $V_i$ which have distance smaller $p$ to $v$ in $S^{*C}$. Applying this observation on distances between $v$ and $V_i$ and on distances between $V_i$ and $V_j$ ($j\not=i$), implies 
$$ |E'|\le 3(2^{p-1}-1)+3(2^{p-1}-1)^2. $$

By the choice of $\C$, $6(2^{p-1}-1)+6(2^{p-1}-1)^2\le \epsilon (\C-1) $. Now $\operatorname{cost}(S^{*C})
 \ge n(S^{*C})-1 \ge \C-1$ implies $2|E'|\le \epsilon \operatorname{cost}(S^{*C})$.


\begin{align*}
    \operatorname{cost}(S')&=\sum\limits_{xy\in E'} \operatorname{dist}_{S'}(x,y)+\sum\limits_{xy\in E(T)\setminus E': x,y\in V(C)} \operatorname{dist}_{S'}(x,y)\\
    &\le \sum\limits_{xy\in E'} \left(\operatorname{dist}_{S^{*C}}(x,y)+2\right)+\sum\limits_{xy\in E(T)\setminus E': x,y\in V(C)} \left( \operatorname{dist}_{S^{*C}}(x,y)+2\right)\\
    &\le 2\cdot |E'|+ \sum\limits_{xy\in E'} \operatorname{dist}_{S^{*C}}(x,y)+\sum\limits_{xy\in E(T)\setminus E': x,y\in V(C)} \left(1+\epsilon\right) \operatorname{dist}_{S^{*C}}(x,y)\\
    &\le (1+2\epsilon) \operatorname{cost}(S^{*C})
\end{align*}

If $v$ has degree 2 in $S^{*C}$, let $v_1,v_2$ be the neighbors of $v$. Remove the edges $vv_1,vv_2$, add a new Steiner vertex $x_1$, and add the edges $vx_1$, $x_1v_1$, and $x_1v_2$. The bounds from the first case stay true.
\end{proof}

\begin{lemma}\label{lemmaA2}Let $C$ be a component of $T'$.
If $d_C(v)>k$, then
$$\operatorname{cost}(S^C) \le (1+2\epsilon)\OPT(C). $$
    
\end{lemma}

\begin{proof}
Since $C_i$ has at most $c$ vertices and $S^C$ contains no Steiner vertices,
$\operatorname{cost}(S^{C_i})\le \C^2$. 
Since every $u_i$ has fewer than $\C$ neighbors in $C_i$,
Lemma~\ref{lemma_subtree_adding_cost} implies

$$\operatorname{cost}(S^C-v) \le \sum\limits_{i=1}^\ell \operatorname{cost}(S^{C_i}) + \C\cdot d_C(v).$$

Since $d_C(v) > k \ge D$, 
$\sum\limits_{i=1}^\ell \operatorname{dist}_{S^C}(v,u_i)\le (1+\epsilon) 
\OPT(K_{1,\ell}).
$
Furthermore,
$$\OPT(C)\ge \OPT(K_{1,\ell}) \ge \alpha d_C(v) \log_2 d_C(v).$$

Together, since $d_C(v)>k$, the choice of $k$ implies
\begin{align*}
\operatorname{cost}(S^C)&\le \C^2\cdot d_C(v)+ \C\cdot d_C(v)+  \sum\limits_{i=1}^\ell \operatorname{dist}_{S_C}(v,u_i)\\
&\le \epsilon \alpha d_C(v) \log_2 d_C(v) + (1+\epsilon) \OPT(K_{1,\ell}) \le (1+2\epsilon) \OPT(C). 
\end{align*}
\end{proof}

\subsection{Analysis of Algorithm 4}

\begin{lemma}\label{lemmaA3}
Let $S'^C$ be the resulting tree when applying Algorithm~\ref{algospace} to $S^C$.
    $$\operatorname{cost}(S'^C)\le (1+4\epsilon) \operatorname{cost}(S^C)$$
\end{lemma}

\begin{proof}
For the readability of the proof, let $S=S^C$ and $S'=S'^C$. Note that $S^C$ can contain Steiner vertices, but $n(S^C)\le 2n(C)-1$. As always, if $x$ is a Steiner vertex of $S^C$, then $x\not\in V(C)$ and we set $d_C(x)=0$.\\
Note that $\{xy\in E(T): x,y\in V(C)\}=E(C)$ since $C$ is a subtree of $T$.
We bound the cost-increase, that is,
$$
\operatorname{cost}(S')-\operatorname{cost}(S)=
\sum\limits_{xy\in E(C)} \left(\operatorname{dist}_{S'}(x,y)-\operatorname{dist}_S(x,y) \right).$$

Therefore we consider a simple charging scheme from the start of the algorithm until the end. 
Each time the distance between the endvertices of an edge 
$xy\in E(C)$
increases throughout the algorithm, let $v$ be the corresponding vertex in that iteration of the algorithm. Potentially after renaming the vertices $x$ and $y$, w.l.o.g. $x \in V(S'_v) $ and $y \not \in V(S'_v)$ and we consider the following cases.
\begin{enumerate}
    \item If $\operatorname{dist}_{S}(x,y)\ge  \lceil \log_2(\C) \rceil$, 
charge the cost-increase (of 1) at the edge $xy$.

    \item If $\operatorname{dist}_{S}(x,y)< \lceil \log_2(\C) \rceil$, consider the following cases.


If $d_C(x)> \Delta$, we charge the cost-increase at the vertex $x$. 

If $d_C(x)\le \Delta$ and $x$ has distance at most $p$ to $v$ in $S'_v$, we charge the cost-increase at $S$.

If $d_C(x)\le \Delta$ and $x$ has distance greater than $p$ to $v$ in $S'_v$, we charge the cost-increase at the edge $xy$.
    
\end{enumerate}

We count the charges. 
Note that we only increase the cost of the solution in an iteration of the algorithm if every vertex in
$S'_v$ with distance less than $\lceil \log_2(\C)\rceil$ to $v$ has out-degree 2, which implies that those vertices form a perfect binary tree and 
there are at least $\C-1$ many vertices with distance less than $\lceil \log_2(\C)\rceil$ to $v$ in $S'_v$. Note that we call a binary tree with $\ell$ vertices perfect if it is complete and $\log_2(\ell+1)\in \mathbb{N}$. \\
We first consider $S$, which only gets charged in the second case.
Consider a vertex $v$, not the root of $S$, whose distance to $a$ is a multiple of $\lceil \log_2(\C) \rceil$ and where the vertices in
$S'_v$ with distance less than $\lceil \log_2(\C)\rceil$ form a perfect binary tree. There are at most $2^{p+1}-1$ vertices which have distance at most $p$ to $v$ in $S'_v$.
There are at most $\frac{n(S)-1}{\C-1}\le \frac{2n(C)-2}{\C-1}$ choices for such a $v$ in $S$.
Together with the degree bound $\Delta$, this implies a charge at $S$ of at most
$$\Delta \left(2^{p+1}-1\right) \frac{2n(C)-2}{\C-1}
<\epsilon (\C-1)\frac{2n(C)-2}{\C-1}=2\epsilon (n(C)-1) \le 2\epsilon \operatorname{cost}(S).
$$

For an edge $xy\in E(T)$, where the vertices $x,y\in V(S)$ have distance at least $\lceil \log_2(\C)\rceil$ in $S$, 
let $i\in \mathbb{N}$ such that $i \lceil \log_2(\C) \rceil  \le dist_{S}(x,y) <  (i+1) \lceil \log_2(\C) \rceil$. Let $q$ be the lowest common ancestor of $x$ and $y$ in $S$. During an iteration of the for-loop, the distance between $x$ and $y$ only increases if $x\in V(S_v)$ and $q\not\in V(S_v)$, or the other way round. There are at most $\left\lfloor \frac{\operatorname{dist_S}(x,q)}{\left\lceil \log_2(\C) \right\rceil} \right\rfloor +1$ (originally) marked vertices $v$ with $x\in V(S_v)$ and $q\not\in V(S_v)$, analogously for $y$. Since $\operatorname{dist_S}(x,q)+\operatorname{dist_S}(y,q)=\operatorname{dist_S}(x,y)$, this implies a charge at $xy$ of at most 
$i+2$.\\

We count the charge at an edge $xy\in E(T)$ with $x,y\in V(C)$ and $p< \operatorname{dist}_S(x,y)<\lceil \log_2(\C) \rceil$. Since $\operatorname{dist}_S(x,y)<\lceil \log_2(\C) \rceil$, throughout the algorithm, there exist at most two marked vertices $v$ such that the distance between $x$ and $y$ increases during the corresponding for-loop.
This implies a charge of at most 2 at $xy$.\\
Edges $xy\in E(T)$ with $x,y\in V(C)$ and $ \operatorname{dist}_S(x,y)\le p$ receive no charge. \\
Since $2<\epsilon p$ and $3< \epsilon \lceil \log_2(\C) \rceil$, the charge at all edges $xy\in E(T)$ with $x,y\in V(T)$ combined is at most
%
\begin{align*}
    &\sum\limits_{\substack{ xy\in E(T): x,y\in V(C) ,\\ p<\operatorname{dist}_S(x,y)<\lceil \log_2(\C) \rceil  }} 2 + \sum\limits_{i=1}^\infty \sum\limits_{\substack{ xy\in E(T): x,y\in V(C) ,\\ i\lceil \log_2(\C) \rceil  \le dist_{S}(x,y) < (i+1)\lceil \log_2(\C) \rceil} }i+2\\[2mm]
&\le \epsilon
\sum\limits_{\substack{xy\in E(T): x,y\in V(C) ,\\ p< dist_{S}(x,y)}}  \operatorname{dist}_S(x,y)
\le \epsilon \operatorname{cost}(S).
\end{align*}

We consider the charge at a vertex $u\in V(C)$ with $d_C(u)>\Delta$.
Note that there is at most one (originally) marked vertex $v$, for which a charge at $u$ can occur in the corresponding for-loop, since only edges $ux\in E(T)$ with $\operatorname{dist}_S(u,x)<\lceil \log_2(\C) \rceil$ potentially induce a charge at $u$. Hence the charge at $u$ is at most $d_C(u)$.\\
The charge at a vertex $u\in V(C)$ with $d_C(u)\le \Delta$ is 0.\\
The charge at all vertices combined is at most 

$$ \sum\limits_{w\in V(C), d_C(w)> \Delta} d_C(w)<\frac{1}{2}\sum\limits_{w\in V(C), d_C(w)> \Delta} \epsilon \alpha d_C(w) \log_2 d_C(w) \le \epsilon \OPT(C) \le \epsilon \operatorname{cost}(S).$$

Combining the bound on the charges of $S$, the edges, and the vertices completes the proof.
\end{proof}

\subsection{Analysis of Algorithm 5}

\begin{lemma}\label{lemmaAA}
Fix an iteration of the while loop and an iteration of the for-loop. Let $S'$ be $S$ at the beginning of Line 5 and let $S''$ be $S$ after Line 10.
$C$, $u$, $u_i$, $d_i$ and $w_i$ are the values of the corresponding variables at this iteration.
$$\operatorname{cost}(S'')\le \operatorname{cost}(S')+\sum\limits_{i=1}^\ell \operatorname{cost}(S^{u_i})
+
    %
    3\epsilon\sum\limits_{i=1}^\ell \operatorname{cost}(S'^{C_i})
    +  \epsilon \sum\limits_{i=1}^\ell \OPT(K_{1,d_i})
    + (1+\epsilon)\OPT(K_{1,\ell})
    $$
    \end{lemma}

\begin{proof}

    We consider the cost-increase when attaching the $S^{u_i}$'s to $f(u)$.
For some $i\in [\ell]$, the distance between $u_i$ and the vertices in $N^+_T(u_i)\cap V(S^{u_i})$ increases in $S^{u_i}$. Lemma \ref{lemma_adding_subtree_cost} implies additional cost of at most $\sum\limits_{i=1}^\ell d_i$. 


If $d_i\le \Delta$, we compare the additional cost of $d_i$ to the cost of $S'^{C_i}$. If $d_i>\Delta$, we compare the additional cost of $d_i$ to $\OPT(K_{1,d_i})$. Since $\Delta < \epsilon (\C-1)$,
\begin{align*}
    \sum\limits_{i=1}^\ell d_i &= \sum\limits_{i\in [\ell]: d_i\le \Delta} d_i + \sum\limits_{i\in [\ell]: d_i>\Delta} d_i\\
    &< \epsilon \sum\limits_{i\in [\ell]:d_i \le \Delta} \operatorname{cost}(S'^{C_i})+\epsilon\sum\limits_{i\in [\ell]: d_i>\Delta} \alpha d_i \log_2 d_i\\
&\le \epsilon \sum\limits_{i\in [\ell]:d_i \le \Delta} \operatorname{cost}(S'^{C_i})+\epsilon\sum\limits_{i\in [\ell]: d_i>\Delta} \OPT(K_{1,d_i}).
\end{align*}


    The cost of the edges $u_iu$ for $i\in [\ell]$ is 
    $$\sum\limits_{i=1}^\ell \operatorname{dist}_{S''}(u,u_i) =\sum\limits_{i=1}^\ell \left(\operatorname{dist}_{S''}(u_i,f(u))+\operatorname{dist}_{S''}(f(u),u)  \right).
    $$
    
Note that
    $\sum\limits_{i=1}^\ell\operatorname{dist}_{S''}(f(u),u) \le \ell (2\lceil \log_2 \C \rceil +1) \le \epsilon \sum\limits_{i=1}^\ell \operatorname{cost}(S'^{C_i})$, since $2\lceil \log_2 \C \rceil +1 \le \epsilon (\C-1)$.

If $\ell \ge D$, then
    $\sum\limits_{i=1}^\ell \operatorname{dist}(u_i,f(u)) \le (1+\epsilon) \OPT(K_{1,\ell})$.

    If $\ell < D$, then $\sum\limits_{i=1}^\ell \operatorname{dist}(u_i,f(u)) \le \ell D \le \epsilon \sum\limits_{i=1}^\ell \operatorname{cost}(S'^{C_i})$, since $D< \epsilon (\C-1)$.

Together, 
$$\sum\limits_{i=1}^\ell \operatorname{dist}_{S''}(u,u_i) 
\le (1+\epsilon) \OPT(K_{1,\ell})+2\epsilon \sum\limits_{i=1}^\ell \operatorname{cost}(S'^{C_i})
    $$

Since $
\operatorname{cost}(S'') \le \operatorname{cost}(S')+\sum\limits_{i=1}^\ell \operatorname{cost}(S^{u_i})+
\sum\limits_{i=1}^\ell d_i+\sum\limits_{i=1}^\ell \operatorname{dist}_{S''}(u,u_i)
$, the above inequalities imply the result.
\end{proof}

\begin{lemma}\label{lemmaA4}
    Let $\epsilon'\ge \epsilon >0$.
If $\operatorname{cost}(S'^{C})\le (1+\epsilon') \OPT(C) $ for every component $C$ of $T'$, then

       $$\operatorname{cost}(S^r)     
    \le (1+\epsilon'+2\epsilon+3\epsilon(1+\epsilon')) \OPT(T),$$
where $r$ is the root of $T$.
\end{lemma} 

\begin{proof} We show by induction that, throughout the algorithm, for $S$ in Line 3 and at the end of Line 10, it holds that


\begin{align}\label{eqcost}
    \operatorname{cost}(S) \le (1+\epsilon') 
\OPT(T[V(S)])
+
\epsilon\sum\limits_{\substack{ v \in V(S) \\ v \text{ root in }T' \\
v \text{ not root of }S}}\OPT\left(K_{1,d^+_{T}(v)}\right) + 3\epsilon \sum\limits_{\substack{C' \text{ component of } T',\\ V(C')\subseteq V(S), C'\not=C}} \operatorname{cost}(S'^{C'}),
\end{align}
where $C$ is the component which was chosen in the corresponding while-loop, that is, $C$ is the component of $T'$ which contains the root of $S$.\\
If $S=S'^{C}$ (Line 3),
then
$ \operatorname{cost}(S) \le (1+\epsilon') \OPT(C)$,
which implies inequality (\ref{eqcost}).\\
Now, fix an iteration of the while-loop and an iteration of the for-loop. Let $S'$ be $S$ at the beginning of Line 5 and let $S''$ be $S$ after Line 10. 
Note that $T[V(S')]$ and $T[V(S'')]$ are indeed trees. $C$, $C_i$, $u$, $u_i$, $d_i$ and $w_i$ are the values of the corresponding variables at this iteration.\\
Let $S^*$ be an optimal solution of $T[V(S'')]$.

By the induction hypothesis, (\ref{eqcost}) is true for $S'$ and all $S^{u_i}$. Let $S^*_0$ be an optimal solution of $T[V(S')]$ and let $S^*_i$ be an optimal solution of $T[V(S^{u_i})]=T_{u_i}$ for $i\in [\ell]$.
Note that $S^*$ is also a solution for $T_{u_i}$ and for $T[V(S')]$.
This implies
$$\OPT(T[V(S')])=\sum\limits_{xy\in E(T): x,y\in V(S')} \operatorname{dist}_{S^*_0}(x,y)
\le \sum\limits_{xy\in E(T): x,y\in V(S')} \operatorname{dist}_{S^*}(x,y)
$$ and
$$\OPT(T_{u_i})=\sum\limits_{xy\in E(T): x,y\in V(S^{u_i})} \operatorname{dist}_{S^*_i}(x,y)
\le \sum\limits_{xy\in E(T): x,y\in V(S^{u_i})} \operatorname{dist}_{S^*}(x,y)
$$

for $i=1,\ldots, \ell$.
Hence, the induction hypothesis implies
\begin{align*}
&\operatorname{cost}(S')+\sum\limits_{i=1}^\ell \operatorname{cost}(S^{u_i})\\
&\le (1+\epsilon')\sum\limits_{\substack{xy\in E(T): x,y\in V(S') \text{ or}\\ x,y\in V(S^{u_i}) \text{ for some }i\in [\ell] }} \operatorname{dist}_{S^*}(x,y)
+\epsilon\sum\limits_{\substack{v\in V(S), \text{ }v \text{ root in }T'\\v \not\in \{a,u_1,\ldots, u_\ell\},\\ }}\OPT\left(K_{1,d^+_{T}(v)}\right) \\
&+ 3\epsilon \sum\limits_{\substack{C' \text{ component of } T',\\ V(C')\subseteq V(S''), C'\not \in \{C,C_1,\ldots, C_\ell\}}} \operatorname{cost}(S'^{C'}), 
\end{align*}

where $a$ is the root of $S'$.
Note that $\left\{xy\in E(T): x,y\in V(S') \text{ or } x,y\in V(S^{u_i}) \text{ for some }i\in [\ell] \right\}$
together with $\{uu_1,\ldots, uu_\ell\}\}$ is exactly $\{xy\in E(T): x,y\in V(S'')\}$.
Hence adding
$$
    3\epsilon\sum\limits_{i=1}^\ell \operatorname{cost}(S'^{C_i})
    +  \epsilon \sum\limits_{i=1}^\ell \OPT(K_{1,d_i})
    + (1+\epsilon)\OPT(K_{1,\ell})
    $$

on both sides, using
$d_i=d^+_T(u_i)$ for $i\in [\ell]$,
$\OPT(K_{1,\ell}) \le \sum\limits_{i\in [\ell]} \operatorname{dist}_{S^*}(u_i,u)$, 
%
and applying Lemma~\ref{lemmaAA} implies (\ref{eqcost}), which completes the induction.\\

Now we show that (\ref{eqcost}) implies the theorem.
Let $S^*$ be an optimal solution of $T$. Since $S^*$ is a solution of $C$,

$$\OPT(C)
\le \sum\limits_{xy\in E(T): x,y\in V(C)} \operatorname{dist}_{S^*}(x,y)
$$
for every component $C$ of $T'$. Hence

$$\sum\limits_{C' \text{ component of } T'} \operatorname{cost}(S'^{C'}) \le \sum\limits_{C' \text{ component of } T'}(1+\epsilon')\OPT(C') \le (1+\epsilon')\OPT(T)$$

Note
$$\sum\limits_{\substack{ v \text{ root in }T', \\
v \text{ not root of }T}}\OPT\left(K_{1,d^+_{T}(v)}\right)
\le \sum\limits_{ v \in V(T)}\OPT\left(K_{1,d^+_{T}(v)}\right)\le 2\OPT(T).
$$

Together, for the root $r$ of $T$,

$$\operatorname{cost}(S^r)   
\le (1+\epsilon') 
\OPT(T)
+
2\epsilon \OPT(T)
 + 3\epsilon (1+\epsilon') \OPT(T),$$
as stated in the lemma.
\end{proof}

After scaling appropriately, Lemmas \ref{lemmaA1}, \ref{lemmaA2}, \ref{lemmaA3}, and \ref{lemmaA4} together show Theorem \ref{thmptas}. This completes the analysis of the PTAS in Section 5.

\newpage

\bibliographystyle{hyperabbrv}
\bibliography{references}

@online{paper,
  title   = {Designing Approximate Binary Trees for Trees},
  author  = {Leon Kellerhals and Mitja Krebs and Andr{\'e} Nichterlein and Stefan Schmid},
  year    = {2026},
  eprint  = {2604.20786},
  eprinttype = {arXiv},
  url     = {https://arxiv.org/abs/2604.20786}
}

@TechReport{repec:ags:eureia:272157,
type={Econometric Institute Archives},
institution={Erasmus University Rotterdam},
author={Johnson, D. S. and Lenstra, J. K. and Rinnooy Kan, A. H. G.},
title={The Complexity Of The Network Design Problem},
year={1977},
number={272157},
doi={10.22004/ag.econ.272157},
url={https://ideas.repec.org/p/ags/eureia/272157.html},
}

@article{wu,
author = {Wu, Bang Ye and Lancia, Giuseppe and Bafna, Vineet and Chao, Kun-Mao and Ravi, Ramamoorthi and Tang, Chuan},
year = {1999},
month = {01},
pages = {761-778},
title = {A Polynomial-Time Approximation Scheme for Minimum Routing Cost Spanning Trees},
volume = {29},
journal = {SIAM J. Comput.},
doi = {10.1145/314613.314628}
}

@article{hu,
author = {Hu, T. C.},
title = {Optimum Communication Spanning Trees},
journal = {SIAM Journal on Computing},
volume = {3},
number = {3},
pages = {188-195},
year = {1974}
}

@article{DAHLHAUS2004255,
title = {A linear-time algorithm to compute a MAD tree of an interval graph},
journal = {Information Processing Letters},
volume = {89},
number = {5},
pages = {255-259},
year = {2004},
issn = {0020-0190},
doi = {https://doi.org/10.1016/j.ipl.2003.11.009},
url = {https://www.sciencedirect.com/science/article/pii/S0020019003005143},
author = {Elias Dahlhaus and Peter Dankelmann and R. Ravi}
}

@article{LIERS201685,
title = {Binary Steiner trees: Structural results and an exact solution approach},
journal = {Discrete Optimization},
volume = {21},
pages = {85-117},
year = {2016},
issn = {1572-5286},
doi = {https://doi.org/10.1016/j.disopt.2016.05.006},
url = {https://www.sciencedirect.com/science/article/pii/S157252861630038X},
author = {Frauke Liers and Alexander Martin and Susanne Pape}
}

@inproceedings{Bartal,
author = {Bartal, Yair},
title = {On approximating arbitrary metrices by tree metrics},
year = {1998},
isbn = {0897919629},
publisher = {Association for Computing Machinery},
address = {New York, NY, USA},
url = {https://doi.org/10.1145/276698.276725},
doi = {10.1145/276698.276725},
booktitle = {Proceedings of the Thirtieth Annual ACM Symposium on Theory of Computing},
pages = {161–168},
numpages = {8},
location = {Dallas, Texas, USA},
series = {STOC '98}
}

@article{FAKCHAROENPHOL2004485,
title = {A tight bound on approximating arbitrary metrics by tree metrics},
journal = {Journal of Computer and System Sciences},
volume = {69},
number = {3},
pages = {485-497},
year = {2004},
note = {Special Issue on STOC 2003},
issn = {0022-0000},
doi = {https://doi.org/10.1016/j.jcss.2004.04.011},
url = {https://www.sciencedirect.com/science/article/pii/S0022000004000637},
author = {Jittat Fakcharoenphol and Satish Rao and Kunal Talwar}
}

@article{dreyfus,
author = {Dreyfus, S. E. and Wagner, R. A.},
title = {The steiner problem in graphs},
journal = {Networks},
volume = {1},
number = {3},
pages = {195-207},
year = {1971}
}

@incollection{karp,
  author    = {Richard M. Karp},
  title     = {Reducibility Among Combinatorial Problems},
  booktitle = {Complexity of Computer Computations},
  editor    = {R. E. Miller and J. W. Thatcher and J. D. Bohlinger},
  series    = {The IBM Research Symposia Series},
  publisher = {Springer},
  year      = {1972},
  pages     = {85--103},
  address   = {Boston, MA},
  doi       = {10.1007/978-1-4684-2001-2_9}
}

@article{CHUNG198443,
title = {On optimal linear arrangements of trees},
journal = {Computers \& Mathematics with Applications},
volume = {10},
number = {1},
pages = {43-60},
year = {1984},
issn = {0898-1221},
author = {F.R.K. Chung}
}

@inproceedings{monien1985complexity,
  title={The complexity of embedding graphs into binary trees},
  author={Monien, Burkhard},
  booktitle={International Conference on Fundamentals of Computation Theory},
  pages={300--309},
  year={1985},
  organization={Springer}
}

@article{mejri2023survey,
  title={A survey on network design problems: main variants and resolution approaches},
  author={Mejri, Imen and Layeb, Safa Bhar and Zeghal, Farah},
  journal={European Journal of Industrial Engineering},
  volume={17},
  number={2},
  pages={253--309},
  year={2023},
  publisher={Inderscience Publishers (IEL)}
}

@article{diaz2002survey,
  title={A survey of graph layout problems},
  author={D{\'\i}az, Josep and Petit, Jordi and Serna, Maria},
  journal={ACM Computing Surveys (CSUR)},
  volume={34},
  number={3},
  pages={313--356},
  year={2002},
  publisher={ACM New York, NY, USA}
}

@techreport{GoldbergKlipker1976,
author = {M. K. Goldberg and I. A. Klipker},
title = {Minimal Placing of Trees on a Line},
institution = {Physico-Technical Institute of Low Temperatures,
Academy of Sciences of the Ukrainian SSR},
year = {1976},
note = {Technical report, USSR (in Russian)}
}

@article{Shiloach,
author = {Shiloach, Yossi},
title = {A Minimum Linear Arrangement Algorithm for Undirected Trees},
journal = {SIAM Journal on Computing},
volume = {8},
number = {1},
pages = {15-32},
year = {1979}
}

\newpage

\appendix 
\section{Linear Run Time Adaptation}
\label{appendix a}

As already mentioned in Remark \ref{remark_linear_runtime}, we can adapt Algorithm \ref{algorithm_trees_to_binary_trees} to get a linear runtime by only dividing the vertices $u_1, \ldots, u_d$ into two sets:
one containing the vertices with more children
than the average, and the other containing those with at most the average number of children.
The adapted algorithm can be seen below. 

\begin{algorithm} 
    \caption{Tree to Binary Tree - Linear Runtime}
    \label{algorithm_trees_to_binary_trees_linear_runtime}
    \renewcommand{\algorithmicrequire}{\textbf{Input:}}
    \renewcommand{\algorithmicensure}{\textbf{Output:}}
    \begin{algorithmic}[1]
    \Require tree $T$
    \Ensure binary tree $S$ with $V(S) = V(T)$ 
        \While{there exists an unprocessed vertex $v \in V(T)$}
            \State \parbox[t]{435pt}{let $v \in V(T)$ be not yet processed with children $u_1, \ldots, u_d$ such that the subtrees \\$S^{u_1},\ldots, S^{u_d}$ have already been build}
            \If{$d=0$}
        \State $S^v\coloneqq (\{v\},\emptyset) $
        \State continue
    \EndIf
            \State create the subtree $S^v_+$ as follows:
            \State \hspace{\algorithmicindent} define $\bar{d}_u^+ \coloneqq \frac{1}{d} \sum_{i = 1}^{d^+(v)} d^+(u_i)$ \Comment{average number of children of $u_1,\ldots, u_d$} \label{algorithm_linear_runtime_definition_average_degree}
            \State \hspace{\algorithmicindent} define the sets $A \coloneqq \setbuild{u_i}{d^+(u_i) > \bar{d}_u^+}$ and $B \coloneqq \setbuild{u_i}{d^+(u_i) \leq \bar{d}_u^+}$
            \State \hspace{\algorithmicindent} delete the edges between $u_i$ and $w_i$ for all $i = 1, \ldots, d$ \label{algorithm_lineartime_creating_Sv+_start}
            \State \hspace{\algorithmicindent} let $v$ be the root of $S^v_+$ with a single child $u_b$, where $u_b$ is some arbitrary vertex in $B$
            \State \hspace{\algorithmicindent} \parbox[t][0pt]{415pt}{starting from $u_b$ add the vertices in the set $B$ such that $(S^v_+)_{u_b}$ is a complete binary\\ tree in $S^v_+$} \label{algorithm_lineartime_creating_Sv+_end}
            \State \hspace{\algorithmicindent} add the vertices $u_a \in A$ to $S^v_+$ such that $(S^v_+)_{u_b}$ remains a complete binary tree
            \ForAll{$u_1, \ldots, u_d$ in the reverse order in which they have been added to $S^v_+$}
                \State let $u_i$ be the vertex currently considered
                \State consider the subtree $(S^v_+)_{u_i}$ of $S^v_+$ which contains $u_i$ and all its descendants in $S^v_+$
                \State add an edge from $w_i$ to some leaf of $(S^v_+)_{u_i}$ which currently has a degree $<3$
            \EndFor
            \State define $S^v$ as the tree created by connecting the subtrees $S^{u_i}-u_i$ to $S^v_+$
        \EndWhile
    \end{algorithmic}
\end{algorithm}

\begin{remark}
    Note that lines 13--16 can be implemented in $O(d)$ time, for example, by using the algorithm described in the proof of Lemma \ref{marble}, with the marble movement performed in depth-first search order.

    Since we do not require the vertices $u_1, \ldots, u_d$ to be sorted w.r.t.\ their degree anymore, but only divide them into the two sets $A$ and $B$, it follows immediately from the observations done in Remark \ref{remark_runtime} that Algorithm \ref{algorithm_trees_to_binary_trees_linear_runtime} has linear runtime.
\end{remark}

\noindent Note that we only need to show that Inequality \eqref{equation_proof_adding_subtree_cost}
\[
\sum_{k = 1}^{\floor{\log_2(d)}} k \sum_{u_i \in L_k(S^v_+)} d^+(u_i) \leq \sum_{k = 1}^{\floor{\log_2(d)}} k \sum_{u_i \in L_k(S^v_+)} \frac{1}{d} \sum_{i=1}^{d} d^+(u_i)
\]
from the proof of Lemma \ref{lemma_adding_subtree_cost} still holds. The rest of the analysis in Sections \ref{section_worst_case_analysis} and \ref{section_convergence_for_large_degrees} remains the same. To this end, let $v \in V(T)$ be some arbitrary vertex with children $u_1, \ldots, u_d$. Reconsider the level sets $L_k(S^v_+)$ defined in Section \ref{section_worst_case_analysis} w.r.t.\ the subtrees $S^v_+$ created in Algorithm \ref{algorithm_trees_to_binary_trees_linear_runtime}. We define the set containing all children of $v$ with a distance of at least $1$ to some leaf in $S^v_+$ by
\[
\bar{L}(S^v_+) \coloneqq \bigcup_{k = 1}^{\floor{\log_2(d)}} L_k(S^v_+).
\]
If $\bar{L}(S^v_+) \subseteq B$, then Inequality \eqref{equation_proof_adding_subtree_cost} holds trivially. If this is not the case, then $L_0(S^v_+) \subseteq A$ and 
\[
\sum_{u_i \in L_0(S^v_+)} d^+(u_i) > \abs{L_0(S^v_+)} \bar{d}_u^+,
\]
where $\bar{d}_u^+ \coloneqq \frac{1}{d} \sum_{i = 1}^{d^+(v)} d^+(u_i)$ as in line \ref{algorithm_linear_runtime_definition_average_degree} of Algorithm \ref{algorithm_trees_to_binary_trees_linear_runtime}. Hence, 
\begin{equation} \label{equation_linear_runtime}
    \sum_{u_i \in \bar{L}(S^v_+)} d^+(u_i) \leq \abs{\bar{L}(S^v_+)} \bar{d}_u^+.
\end{equation}
From Inequality \eqref{equation_linear_runtime}, it follows that we can charge for all $u_i \in A \cap \bar{L}(S^v_+)$ the excess value $d^+(u_i) - \bar{d}_{u}^+$ to the value $d^+(u_j)$ of some vertex (or potentially split it and charge it to multiple vertices) $u_j \in B$, such that after the charging $d^+(u_j) \leq \bar{d}_{u}^+$ still holds.\\ 
By the construction of $S_v$ in lines \ref{algorithm_lineartime_creating_Sv+_start}-\ref{algorithm_lineartime_creating_Sv+_end} of Algorithm \ref{algorithm_trees_to_binary_trees_linear_runtime}, all vertices in the set $B$ have been added to $S^v_+$ first and are hence on a higher or at least the same level as the vertices in $A$. The charging of the excess values to the vertices in $B$ can only increase \eqref{equation_proof_adding_subtree_cost} on the l.h.s.\ and we can conclude that Inequality \eqref{equation_proof_adding_subtree_cost} and consequently Lemma \ref{lemma_adding_subtree_cost} also hold when adding the subtrees $S^{u_i}-u_i$ to $S^v_+$ in Algorithm \ref{algorithm_trees_to_binary_trees_linear_runtime}.

\end{document}